\documentclass[11pt]{article}

\usepackage{amsmath,amssymb,amsthm,color,mathrsfs,amsbsy}
\usepackage{graphicx}
\usepackage{psfrag}
\newtheorem{theorem}{\bf Theorem}
\newtheorem{lemma}[theorem]{\bf Lemma}
\newtheorem{proposition}[theorem]{\bf Proposition}

\newtheorem{condition}[theorem]{\bf Condition}
\newtheorem{problem}[theorem]{\bf Problem}

\renewcommand{\Re}{\mathrm{Re}}
\renewcommand{\Im}{\mathrm{Im}}

\newcommand\RR{\mathbb{R}}
\newcommand\CC{\mathbb{C}}
\newcommand\ZZ{\mathbb{Z}}

\newcommand\ep{\epsilon}
\newcommand\ka{\tilde\kappa}
\newcommand\om{\tilde\omega}
\newcommand\kw{(\kappa,\omega)}
\newcommand\kwz{(\kappa_0,\omega_0)}

\newcommand\zwz{(0,\omega_0)}
\newcommand{\usc}{u^\text{\scriptsize sc}}
\newcommand{\uinc}{u^\text{\scriptsize inc}}
\newcommand{\winc}{w^\text{\scriptsize inc}}
\newcommand{\DtN}{N}

\newcommand\slantfrac[2]{\hspace{3pt}\!^{#1}\!\!\hspace{1pt}/ \hspace{2pt}\!\!_{#2}\!\hspace{3pt} }
\newcommand{\onehalf}{{\slantfrac{1}{2}}}
\newcommand{\half}{\text{\tiny $\textstyle{\frac{1}{2}}$}}
\newcommand{\mhalf}{\text{\tiny --$\textstyle{\frac{1}{2}}$}}
\newcommand{\Honeper}{{H^1_{\!\text{per}}(\Omega)}}

\newcommand{\mat}[5]{ \renewcommand{\arraystretch}{#1}
                    \left[\! \begin{array}{cc}
                            #2 & #3 \\
                            #4 & #5 \end{array} \!\right] }

\begin{document}

\begin{center}
{\bf \Large Total Resonant Transmission and Reflection}
\\ \vspace{1ex}
{\bf \Large by Periodic Structures}
\end{center}

\vspace{0.2ex}

\begin{center}
{\scshape \large Stephen P. Shipman, Hairui Tu}
\end{center}

\begin{center}
{\itshape Department of Mathematics, Louisiana State University\\
Baton Rouge, LA \ 70803}
\end{center}

\vspace{3ex}
\centerline{\parbox{0.9\textwidth}{
{\bf Abstract.}\
Resonant scattering of plane waves by a periodic slab under conditions close to those that support a guided mode is accompanied by sharp transmission anomalies.  For two-dimensional structures, we establish sufficient conditions, involving structural symmetry, under which these anomalies attain total transmission and total reflection at frequencies separated by an arbitrarily small amount.  The loci of total reflection and total transmission are real-analytic curves in frequency-wavenumber space that intersect quadratically at a single point corresponding to the guided mode.  A single anomaly or multiple anomalies can be excited by the interaction with a single guided mode.
}}

\vspace{3ex}
\noindent
\begin{mbox}

\end{mbox}

\vspace{3ex}
\hrule
\vspace{2ex}

\section{Introduction}

A dielectric slab with periodically varying structure can act both as a guide of electromagnetic waves as well as an obstacle that diffracts plane waves striking it transversely.  Ideal guided modes exponentially confined to a slab are inherently nonrobust objects, tending to become leaky, or to radiate energy, under perturbations of the structure or wavevector.  The leaky modes are manifest as high-amplitude resonant fields in the structure that are excited during the scattering of a plane wave.  One may think of this phenomenon loosely as the resonant interaction between guided modes of the slab and plane waves originating from exterior sources.  This interaction generates sharp anomalies in the transmittance, that is, in the fraction of energy of a plane wave transmitted across the slab as a function of frequency and wavevector.

In this work, we analyze a specific feature of transmission resonances for two-dimensional
lossless periodic structures (Fig.~\ref{fig:Slab}) that results from perturbation of the wavenumber from that of a true (exponentially confined to the structure) guided mode.  Graphs of transmittance \textit{vs.}
frequency $\omega$ and wavenumber $\kappa$ parallel to the slab typically exhibit a sharp peak
and dip near the parameters $\kwz$ of the guided mode.
Often in computations these extreme points appear to reach 100\% and 0\% transmission, which means that, between two closely spaced frequencies, the slab transitions from being completely transparent to being completely opaque (Fig.~\ref{fig:RodsTransmission}).
The main result, presented in Theorem \ref{Thm:Main}, is a proof that, if the slab is symmetric with respect to a line parallel to it (the $x$-axis in Fig.~\ref{fig:Slab}), then these extreme values are in fact attained.  Subject to technical conditions discussed later on, it can be paraphrased like this:

\smallskip
\noindent
{\bfseries Theorem.}\;
{\itshape Consider a two-dimensional lossless periodic slab that is symmetric about a line parallel to it.  If the slab admits a guided mode at an isolated wavenumber-frequency pair $\kwz$,
then total transmission and total reflection are achieved at nearby frequencies whose difference tends to zero as $\kappa$ tends to $\kappa_0$.  The loci in real $\kw$-space of total transmission and total reflection are real analytic curves that intersect tangentially at $\kwz$.
}
\smallskip

In this Theorem, it is important that the slab admits a guided mode at an {\em isolated} pair $\kwz$.  The frequency $\omega_0$ is above cutoff, meaning that it lies above the light cone in the first Brillouin zone of wavenumbers~$\kappa$ and therefore in the $\kw$ regime of scattering states (Fig.~\ref{fig:Diamond}).  Perturbing $\kappa$ from $\kappa_0$ destroys the guided mode and causes resonant scattering of plane waves at frequencies near $\omega_0$.  In the literature, one encounters these nonrobust guided modes at $\kappa_0=0$, that is, they are standing waves.  Although we are not aware of truly traveling guided modes  ($\kappa\not=0$) above cutoff in periodic photonic slabs, we believe that they should exist in anisotropic structures.

Resonant transmission anomalies are well known in a wide variety of applications in electromagnetics
and other instances of wave propagation, and a veritable plenitude of models and techniques has been
developed for describing and predicting them \cite{Martin-MorGarcia-VidLezec2001,EbbesenLezecGhaemi1998,LiuLalanne2008,MedinaMesaMarques2008}.  The causes of anomalies are manifold and include Fabry-Perot resonance and Wood anomalies near cutoff frequencies of Rayleigh-Bloch diffracted waves.  The present study addresses the specific resonant phenomenon associated with the interaction of plane waves with a guided mode of a slab structure.

In our point of view, one begins with the equations of electromagnetism (or acoustics, \textit{etc.}),
admitting no phenomenological assumptions that cannot be proved from them, and seeks to provide rigorous theorems on the phenomenon of resonant scattering.
A rigorous asymptotic formula for transmission anomalies in the case of perturbation of the angle of incidence (or Bloch wavevector) has been obtained through singular complex analytic perturbation of the scattering problem about the frequency and wavenumber of the guided mode for two-dimensional structures \cite{ShipmanVenakides2005,PtitsynaShipmanVenakides2008,Shipman2010}, and the analysis in this paper is based on that work.  The essential new result is Theorem~\ref{Thm:Main} on total reflection and transmission.  Previous analyticity results were based on boundary-integral formulations of the scattering problem, which are suitable for piecewise constant scatterers.  Here, we deal with general positive, coercive, bounded coefficients and thus give self-contained proofs of analyticity of a scattering operator based on a variational formulation of the scattering problem.

There are interesting open questions concerning the detailed nature of transmission resonances.
In passing from two-dimensional slabs (with one direction of periodicity) to three-dimensional
slabs (with two directions of periodicity), both the additional dimension of the wavevector parallel
to the slab as well as various modes of polarization of the incident field impart
considerable complexity to the guided-mode structure of the slab and its interaction with plane waves.
The role of structural perturbations is a mechanism for initiating coupling between guided
modes and radiation \cite{FanJoannopoul2002}
\cite[\S4.4]{HaiderShipmanVenakides2002} that also deserves a rigorous mathematical treatment.
A practical understanding of the correspondence between structural parameters and salient features of transmission anomalies, such as central frequency and width, would be valuable in applications.

\smallskip
The main theorem is proved in Section~\ref{Section:Single} in the simplest case in which the transmittance graph exhibits a single sharp peak and dip.  The proof rests on the complex analyticity with respect to frequency and wavenumber inherent in the problem of scattering of harmonic fields by a periodic slab.  The framework for our analysis is the variational (or weak-PDE) formulation of the scattering problem, which is reviewed in Section~\ref{Section:Background}.  Section~\ref{Section:Multiple} deals briefly with non-generic cases in which degenerate or multiple anomalies emanate from a single guided-mode frequency.  A number of graphs of transmittance in the generic and non-generic cases are shown in Section~\ref{Section:Graphs}.


\section{Background: Scattering and guided modes}\label{Section:Background}

Readers familiar with variational formulations of scattering problems can easily skim this section and proceed to Section \ref{Section:Single}, which contains the main result.

A two-dimensional periodic dielectric slab (Fig.~\ref{fig:Slab}) is characterized by two coefficients $\ep(x,z)$
and $\mu(x,z)$, $(x,z)\in\RR^2$, that are periodic in the $x$-direction and constant outside of a strip parallel to the $x$ axis.  We take these coefficients to be bounded from below and above by positive numbers:
\begin{equation}\label{structure}
  \renewcommand{\arraystretch}{1.3}
\left.
  \begin{array}{lll}
    \ep(x+2\pi n,z) = \ep(x,z), & \mu(x+2\pi n,z) = \mu(x,z), & \text{for all }n\in\ZZ, \\
    \ep(x,z) = \ep_0, & \mu(x,z) = \mu_0, & \text{if } |z| \geq L, \\
   0< \ep_- < \ep(x,z) < \ep_+, & 0< \mu_- < \mu(x,z) < \mu_+. &
  \end{array}
\right.
\end{equation}

\begin{figure}  
\centering
\scalebox{0.25}{\includegraphics{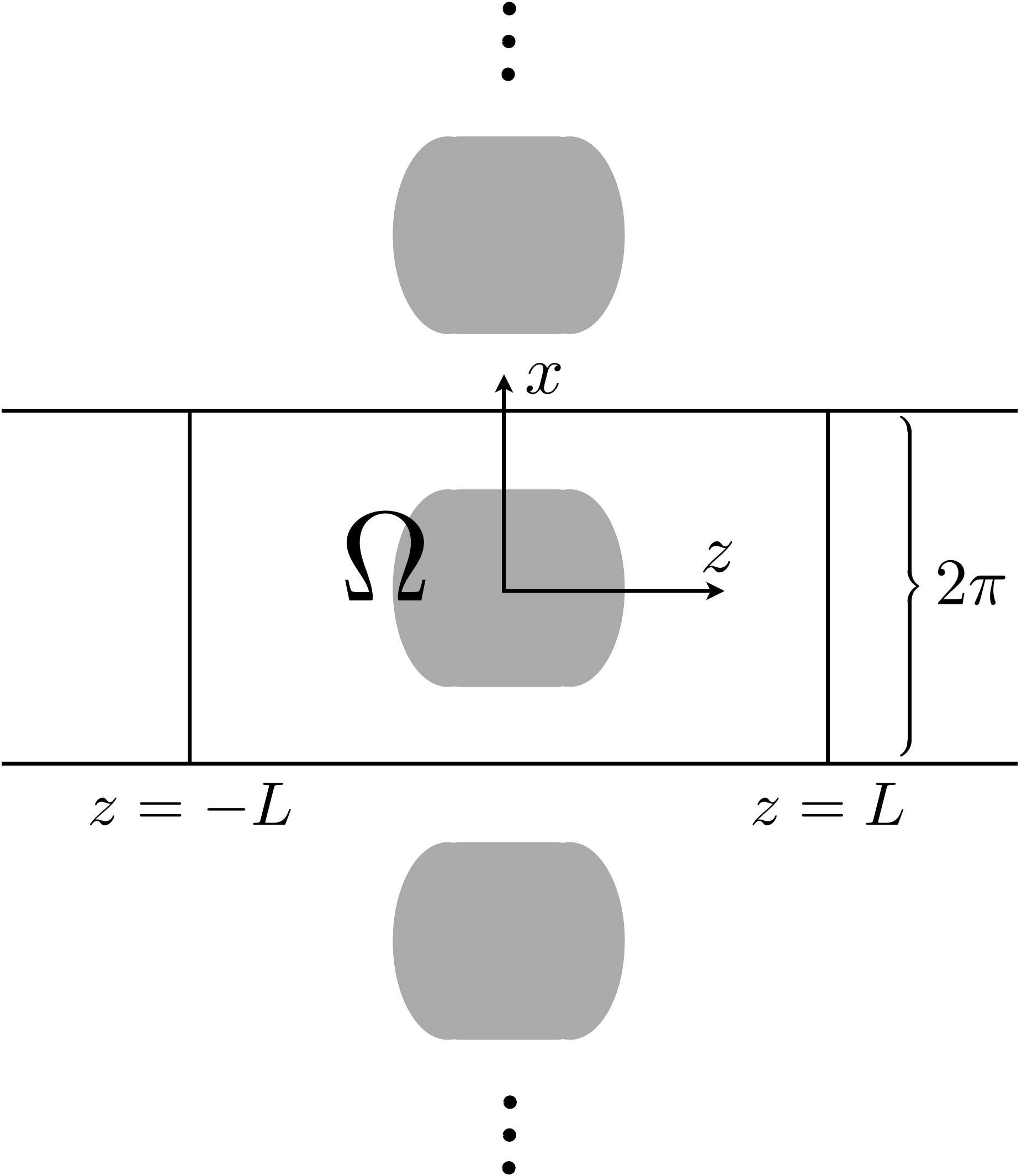}}
\caption{\footnotesize An example of a two-dimensional periodic slab.  One period truncated to the rectangle $[-\pi,\pi]\times[-L,L]$ is denoted by~$\Omega$.}
\centering
\label{fig:Slab}
\end{figure}

\begin{figure}
\centering
\scalebox{0.5}{\includegraphics{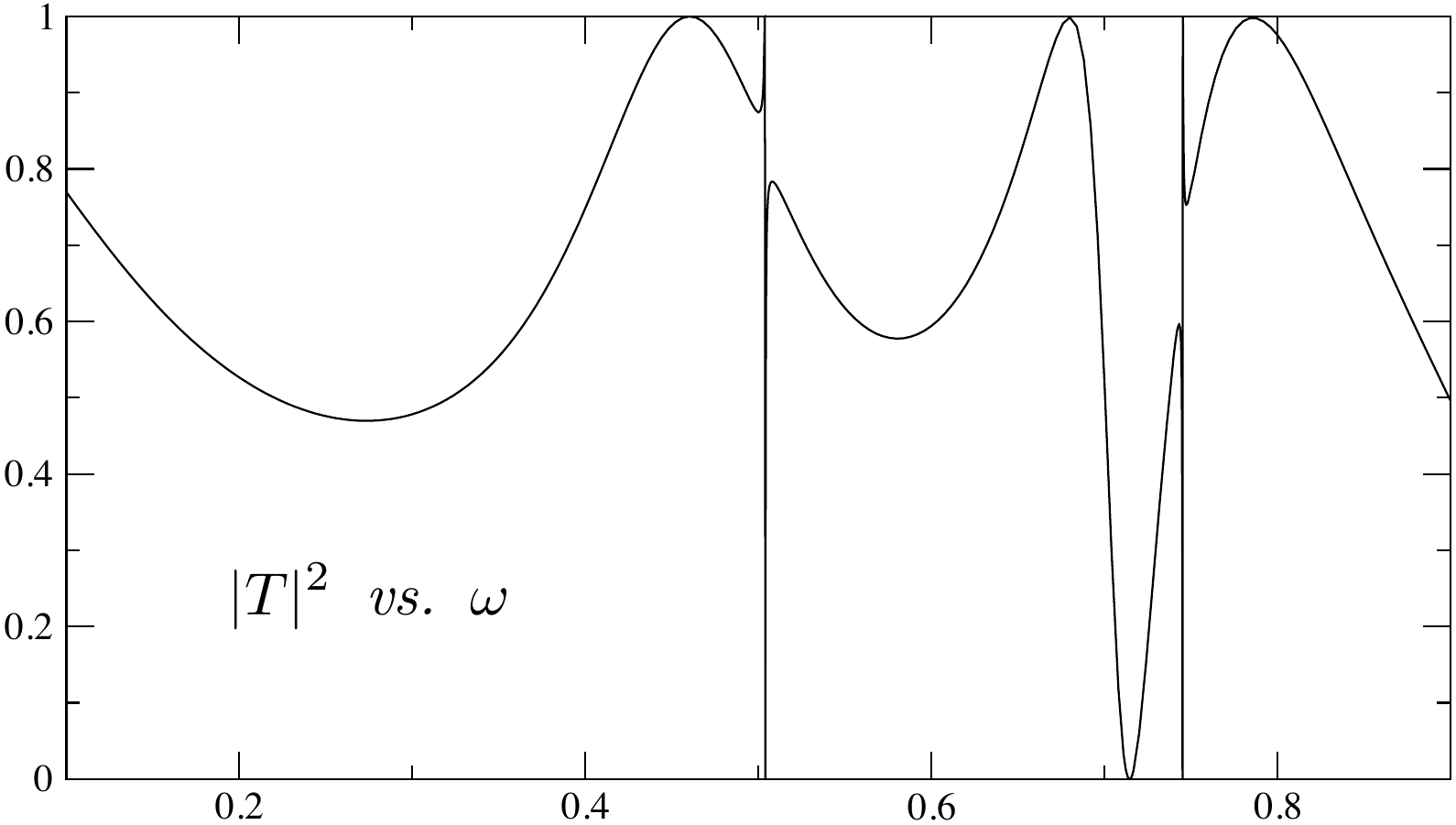}}
\caption{\footnotesize Numerical computation of the transmittance, or the percentage of energy transmitted across a penetrable waveguide of period $2\pi$ as a function of the frequency of the incident plane wave.  Here, the wavenumber in the $x$-direction (Fig.~\ref{fig:Slab}) is $\kappa=0.02$ and one period consists of a single circle of radius $\pi/2$ with $\epsilon=10$ and an ambient medium with $\epsilon=1$; $\mu=1$ throughout.  The structure supports guided modes at $\kw=(0,0.5039...)$ and $\kw=(0,0.7452...)$, both contained within the region $\cal D$ of one propagating diffractive order (Fig~\ref{fig:Diamond}).  Theorem \ref{Thm:Main} guarantees that the transmittance attains minimal and maximal values of 0\% and 100\% at each of the sharp anomalies near the guided-mode frequencies.  This example was computed using the boundary-integral method described in \cite{HaiderShipmanVenakides2002}.}
\centering
\label{fig:RodsTransmission}
\end{figure}

Physically, the structure is three-dimensional but invariant in the $y$-direction.
The Maxwell system for $y$-independent electromagnetic fields in such a structure decouples
into two polarizations and simplifies to the scalar wave equation for the out-of-plane components
of the $E$ field and the $H$ field independently.
We will consider harmonic fields, whose circular frequency~$\omega$ will always be taken to be positive.
Given a frequency $\omega$, plane waves and guided modes are characterized by their propagation constant $\kappa$ in the direction parallel to the slab.  The $y$ component $E_y$ of a harmonic $E$-polarized field with propagation constant $\kappa$ is of the pseudoperiodic form
\begin{eqnarray*}
  && E_y(x,z,t)=u(x,z)e^{i(\kappa x-\omega t)}, \\
  && u(x+2\pi n,z) = u(x,z) \;\;\text{for all }n\in\ZZ,
\end{eqnarray*}
in which the periodic factor $u$ satisfies the equation
\begin{equation}
(\nabla+i\boldsymbol{\kappa})\cdot\mu^{-1}(\nabla+i\boldsymbol{\kappa})u(x,z)+\epsilon\,\omega^2u(x,z)=0,
\end{equation}
with $\boldsymbol\kappa=(\kappa,0)$.
The number $\kappa$ can be restricted to lie in the Brillouin zone $\left[-\onehalf,\onehalf\right)$.
As long as the quantities
\begin{equation}
\eta_m^2=\epsilon_0\mu_0\omega^2-(m+\kappa)^2
\end{equation}
are nonzero for all integers $m$, the general solution of this equation admits a Fourier expansion on each side of the slab,
\begin{equation}\label{Fourier}
  u(x,z) = \sum_{m=-\infty}^\infty (A^\pm_m e^{i\eta_m z}+B^\pm_m e^{-i\eta_m z})e^{imx} \quad \mbox{for } \pm z>L,
\end{equation}
For real $\kappa$ and $\omega>0$, the square root is chosen with a branch cut on the negative imaginary axis, and the sign is taken such that $\eta_m=|\eta_m|$ if $\eta_m^2>0$ and $\eta_m=i|\eta_m|$ if $\eta_m^2<0$.

\subsection{Scattering and guided modes in periodic slabs}

The $L^2$ theory of guided modes underlies the analysis of resonance in Section~\ref{Section:Single}.  We present the pertinent elements of this theory here and refer the reader to more in-depth discussions in \cite{Bonnet-BeStarling1994,Shipman2010}.

In the problem of scattering of plane waves for real $\kw$, one takes in \eqref{Fourier} $A^-_m=B^+_m=0$ for the infinite set of $m$ such that $\eta^2_m<0$ to exclude fields that grow exponentially as $|z|\to\infty$.  The exponentially decaying Fourier harmonics are known as the evanescent diffraction orders.  The finitely many propagating diffractive orders ($\eta^2_m>0$) express the sum of the incident plane waves and the scattered field far from the slab.  In view of the factor $e^{-i\omega t}$ and the convention $\omega>0$, we see that $A^-_m$ and $B^+_m$ are the coefficients of inward traveling plane waves and $A^+_m$ and $B^-_m$ are the coefficients of outward traveling waves.  The linear orders $\eta_m=0$ correspond to ``grazing incidence", and will not play a role in the present study.

If for a given pair $\kw$, $\eta_m\not=0$ for all $m$, the numbers $\eta_m$ can be continued as analytic functions
in a complex neighborhood of $\kw\in\CC^2$.  The following outgoing condition is central to the mathematical
formulation of the scattering problem and the definition of generalized guided modes.
\begin{condition}[Outgoing Condition]\label{cond:outgoing}
A pseudo-periodic function $\tilde u(x,z)=u(x,z)e^{i\kappa x}$ is said to satisfy the {\em outgoing condition}
for the complex pair $\kw$, with $\Re(\omega)>0$ if there exist a real number $L$ and complex coefficients
$\left\{ c^\pm_m \right\}$ such that
\[
u(x,z)=\sum_{m\in\mathbb{Z}}c_m^{\pm}e^{\pm i\eta_mz}e^{imx} \quad \mbox{for } \pm z>L.
\]
\end{condition}

\begin{figure}  
\includegraphics[scale=0.27]{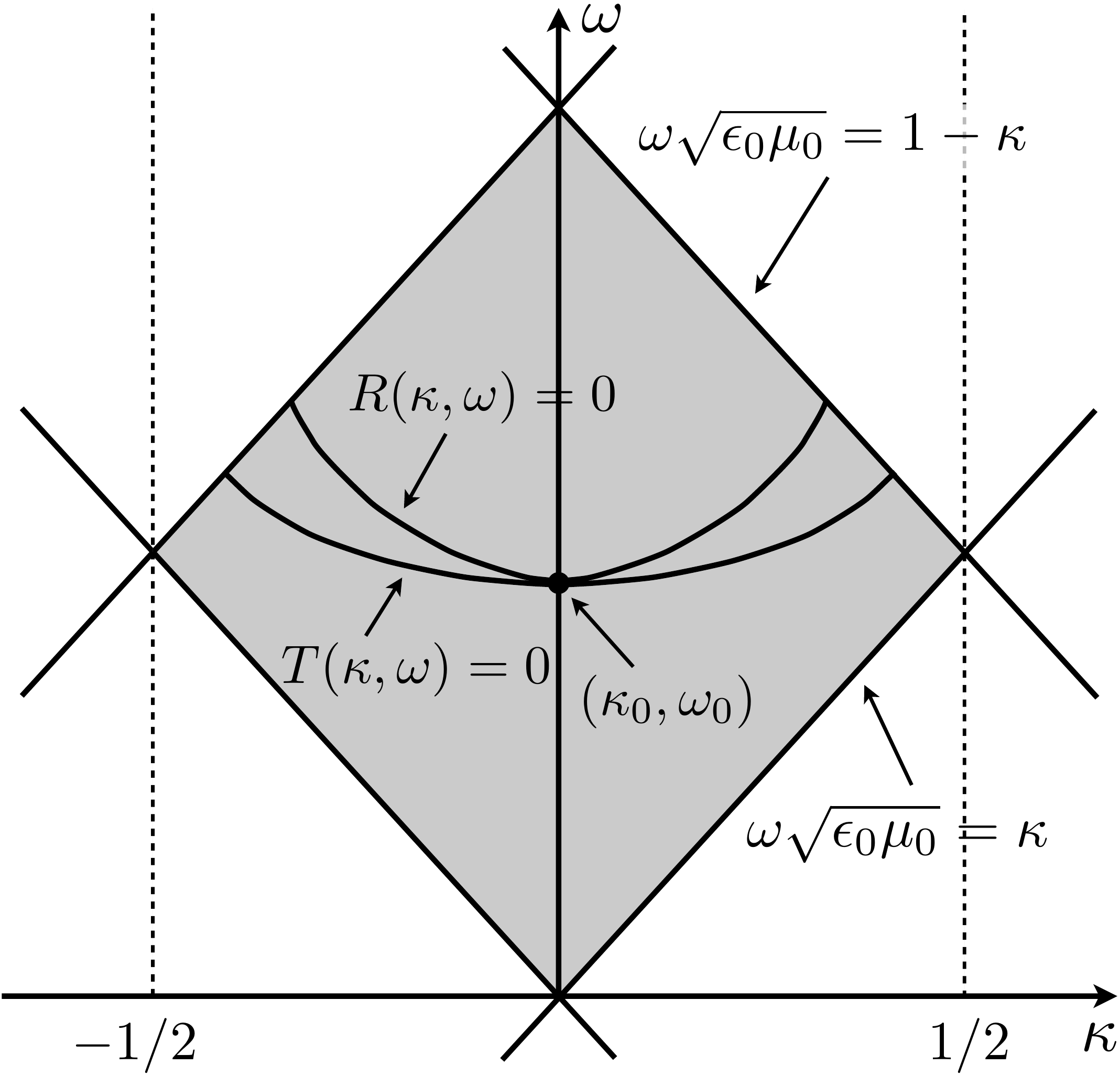}
 \centering
    \caption{\footnotesize The diamond ${\cal D}$ of one propagating diffractive order ($m=0$) within the first Brillouin zone is shown in grey.  The point $\kwz$ represents an isolated real pair on the complex dispersion relation for generalized guided modes.
    If the periodic slab is symmetric about a line parallel to it, as shown in Fig.~\ref{fig:Slab}, the main theorem (Theorem~\ref{Thm:Main}, Sec.~\ref{Section:Single}) guarantees that, under certain generic conditions, the zero sets of transmission, $T\kw = 0$ (or $b\kw=0$), and reflection, $R\kw = 0$ (or $a\kw=0$), intersect quadratically at the pair $\kwz$, where $T$ and $R$ are the complex transmission and reflection coefficients.  Since $|T|^2 + |R|^2=1$, one sees that $|T|$, being continuous in a punctured neighborhood of $\kwz$, achieves all values between $0$ and $1$ in each neighborhood of $\kwz$.  Each zero set either intersects the boundary of ${\cal D}$, as depicted here, or attains an infinite slope.}
    \label{fig:Diamond}
\end{figure}

We will be concerned with the case of exactly one propagating harmonic $m=0$.
This regime corresponds to real pairs $\kw$ that lie in the diamond
\begin{equation*}
  {\cal D} = \left\{ \kw \in\RR^2 : |\kappa|<\onehalf \text{ and } |\kappa| < \omega\sqrt{\ep_0\mu_0} < 1-|\kappa| \right\}
\end{equation*}
shown in Fig.~\ref{fig:Diamond}.
The numbers $\eta_m$ are analytic functions of $\kw$ in a complex neighborhood ${\cal D'}$ of ${\cal D}$; thus, $\RR^2\supset{\cal D}\subset{\cal D'}\subset\CC^2$.

The problem of scattering of plane waves by a periodic slab is the following.
\begin{problem}[Scattering problem]\label{problem:scattering}
Find a function $\tilde u(x,z)$ such that
\begin{equation}\label{outgoing}
\renewcommand{\arraystretch}{1.3}
\left\{
  \begin{array}{l}
\tilde u(x,z) = e^{i\kappa x}u(x,z), \quad\text{$u$ is $2\pi$-periodic in $x$},  \\
(\nabla+i\boldsymbol{\kappa})\!\cdot\!\mu^{-1}(\nabla+i\boldsymbol{\kappa})u(x,z)+\ep\,\omega^2\,u(x,z) =0,   \;\;\mbox{ for } (x,z)\in\mathbb{R}^2,\\
u(x,z) = \uinc(x,z) + \usc(x,z), \quad\mbox{$\usc$ satisfies the outgoing condition,}
    \end{array}
\right.
\end{equation}
in which $\uinc(x,z) = A_0e^{i\eta_0z}+B_0e^{-i\eta_0z}$.
\end{problem}
\noindent

The weak formulation of the scattering problem is posed in the truncated period
\begin{eqnarray*}
  && \Omega=\{(x,z):-\pi<x<\pi,\, -L<z<L \}, \\
  && \Gamma=\bigcup\Gamma_\pm, \quad \Gamma_\pm =\{(x,z):-\pi<x<\pi,\,z=\pm L\}.
\end{eqnarray*}
and makes use of the Dirichlet-to-Neumann map $\DtN=\DtN(\kappa,\omega)$ on the right and left boundaries $\Gamma_\pm$ that characterizes outgoing fields.  This is a bounded linear operator from $H^{\half}(\Gamma)$ to $H^{\mhalf}(\Gamma)$ defined as follows.  For any $f\in H^{\frac{1}{2}}(\Gamma)$, let $\hat f_m=(\hat f_m^+,\hat f_m^-)$ be
the Fourier coefficients of $f$, that is, $f(\pm L,x)=\sum_m\hat f_m^{\pm}e^{imx}$.
Then
\begin{equation}\label{DtN}
\DtN:H^{\frac{1}{2}}(\Gamma)\rightarrow H^{-\frac{1}{2}}(\Gamma),\quad
\widehat{(\DtN f)}_m=-i\eta_m\hat f_m.
\end{equation}
This operator has the property that
\begin{equation*}
  \partial_n u + \DtN u = 0 \;\text{ on }\; \Gamma \;\iff\; \text{$u$ is outgoing,}
\end{equation*}
where ``on $\Gamma$" refers to the traces of $u$ and its normal derivative on $\Gamma$.
In the periodic Sobolev space
\begin{equation*}
  \Honeper = \{u\in H^1(\Omega): u(\pi,z) = u(-\pi,z) \;\text{ for all } z\in(-L,L)\},
\end{equation*}
in which evaluation on the boundaries of $\Omega$ is in the sense of the trace, define the forms
\begin{eqnarray}
  && h(u,v) = h_{\kappa,\omega}(u,v) =\int_{\Omega}\mu^{-1}(\nabla+i\boldsymbol{\kappa})u\cdot(\nabla-i\boldsymbol\kappa)\bar v
                             \,+ \,\mu_0^{-1}\!\!\int_{\Gamma}(\DtN u)v, \label{a} \\
  && b(u,v)=\int_{\Omega}\!\ep\, u \bar v, \notag \\
  && p(v) = p_{\kappa,\omega}(v) =\mu_0^{-1}\!\!\int_{\Gamma}(\partial_n\uinc+\DtN\uinc)\bar v. \notag
\end{eqnarray}

\medskip
\noindent
{\em Remark on notation.}\ The forms $h$ and $p$ depend on $\kappa$ and $\omega$ through $N=N\kw$.  In the sequel, the dependence on $\kappa$ and $\omega$ of certain objects such as $h$, $p$, $N$, $\eta_m$, as well as the operators $A$ and $C_i$ defined below, will be often suppressed to simplify notation.

\begin{problem}[Scattering problem, variational form]\label{prob:weak1}
Given a pair $\kw$, find a function $u\in\Honeper$ such that
\begin{equation}\label{scatteringweak1}
  h(u,v)-\omega^2b(u,v)=p(v), \;\text{ for all } v\in\Honeper.
\end{equation}
\end{problem}

By definition, a {\em generalized guided mode} is a nonvanishing solution of
Problem \ref{prob:weak1} with $p$ set to zero.  Such a solution possesses no
incident field and therefore satisfies the outgoing Condition~\ref{cond:outgoing}.
If $\kw$ is a real pair, then all propagating harmonics in the Fourier expansion
\eqref{Fourier} of the mode must vanish and the field is therefore a true guided mode, which falls
off exponentially with distance from the slab.  This can be proved by integration by
parts, which yields a balance of incoming and outgoing energy flux.
In the case of one propagating harmonic, $\kw\in{\cal D}$, this means
\begin{equation*}
  |A^-_0|^2 + |B^+_0|^2 = |A^+_0|^2 + |B^-_0|^2.
\end{equation*}
Guided modes with $\omega\not=0$ are fundamental in the theory of leaky modes
\cite{PaddonYoung2000,TikhodeevYablonskiMuljarov2002,PengTamirBertoni1975,HausMiller1986} and always are exponentially growing as $|z|\to\infty$ and decaying in time.  The following theorem is proved in \cite[Thm.~5.2]{ShipmanVenakides2003} and \cite[Thm.~15]{Shipman2010}.

\begin{theorem}[Generalized modes] \label{theorem:GenGM}
Let $u$ be a generalized guided mode with $\kappa\in\RR$ (and $\Re(\omega)>0$).
Then $\Im(\omega)\le0$; and $u\rightarrow0$ as $|z|\rightarrow\infty$
if and only if $\omega$ is real.
\end{theorem}

It is convenient to write the form $a-\omega^2b$ as
\[
h(u,v)-\omega^2b(u,v)=c_1(u,v)+c_2(u,v),
\]
in which $c_1(u,v)= h(u,v) + b(u,v)$ and $c_2(u,v)=-(\omega^2+1)b(u,v)$.
Both $c_1$ and $c_2$ are bounded forms in $\Honeper$.
If we take ${\cal D}'$ to be a sufficiently small complex neighborhood of the
diamond ${\cal D}$, $c_1$ is coercive for all $\kw$ in ${\cal D}'$.
These forms are represented by bounded operators $C_1$ and $C_2$ in $\Honeper$:
\[
(C_1u,v)_\Honeper=c_1(u,v),\]
\[
(C_2u,v)_\Honeper=c_2(u,v).
\]
If we denote by $\winc$ the unique element of $\Honeper$ such that
$(\winc,v)_{\Honeper}=p(v)$, the scattering problem becomes
$(C_1u,v)+(C_2u,v)=(\winc,v)$ for all $v\in\Honeper$, or
\begin{equation}\label{operatorform1}
  C_1u+C_2u=\winc.
\end{equation}
Because of the coercivity of $c_1$ and the compact embedding of $\Honeper$ into $L^2(\Omega)$, we have

\begin{lemma}
The operator $C_1$ has a bounded inverse and $C_2$ is compact.
\end{lemma}

By means of the Fredholm alternative one can demonstrate that, even if a slab admits a guided mode
for a given real pair $\kw$, the problem of scattering of a plane wave always has a solution.
Proofs are given in \cite[Thm.~3.1]{Bonnet-BeStarling1994} and \cite[Thm.~9]{Shipman2010}; the idea
is essentially that plane waves contain only propagating harmonics whereas guided modes contain only
evanescent harmonics and are therefore orthogonal to any plain wave source field.

\begin{theorem} \label{theorem:ExistenceSc}
For each $\kw\in{\cal D}$, the scattering Problem \ref{prob:weak1} with a plane-wave source field has at least one solution and the set of solutions is an affine space of finite dimension equal to the dimension of the space of generalized guided modes. The far-field behavior of all solutions is identical.
\end{theorem}

With the notation
\begin{eqnarray*}
  && A\kw = I+C_1\kw^{-1}C_2\kw, \\
  && \psi=u \; \text{ and } \; \phi=C_1^{-1}\winc,
\end{eqnarray*}
equation \eqref{operatorform1} can be
written as
\begin{equation}\label{operatorform2}
A(\kappa,\omega)\psi(\kappa,\omega)=\phi(\kappa,\omega),
  \quad  \text{(Scattering problem in operator form)}
\end{equation}
in which $A$ is the identity plus a compact operator.
A generalized guided mode is a nontrivial solution of the homogeneous problem, in which $\phi$ is set to zero:
\begin{equation}
A(\kappa,\omega)\psi(\kappa,\omega)=0.
  \quad  \text{(Guided mode)}
\end{equation}

It can be proved that, if $\ep$ and $\mu$ are large enough in the structure and symmetric in the $x$ variable
(i.e., about the $z$-axis normal to the slab), there exists a guided mode at some point $\zwz$ in the diamond~${\cal D}$~\cite[\S5.1]{Bonnet-BeStarling1994}.  Such a guided mode is antisymmetric about the $z$-axis and its existence is proven through the decomposition of the operator $A(0,\omega)$ into its action on the subspaces of functions that are symmetric or antisymmetric with respect to $x$.  The symmetry of $A$ is broken when $\kappa$ is perturbed from zero with the consequent vanishing of the guided mode.

The frequency $\omega_0$ should be thought of as an embedded eigenvalue of the pseudo-periodic Helmholtz operator in the strip $\left\{ (z,x): -\pi<x<\pi \right\}$ at $\kappa=0$ which dissolves into the continuous spectrum $\omega\geq|\kappa|/\sqrt{\ep_0\mu_0}$ as $\kappa$ is perturbed.  It is the nonrobust nature of this eigenvalue that is responsible for the resonant scattering and transmission anomalies that we study in this paper.



\subsection{Analyticity}

Analysis of scattering resonance near the parameters $\kwz$ of a guided mode rests upon the analyticity of the operator $A$.  The proof of analyticity is in the Appendix.

\begin{lemma}\label{lemma:analyticityC1}
The operators $C_1$, $C_2$, and $A$ are analytic with respect to $\omega$ and $\kappa$, as long as
$\eta_m^2 = \omega^2\epsilon_0\mu_0-(\kappa+m)^2 \not=0$ for all $m\in\ZZ$, which holds in particular for $\kw\in{\cal D}$.
\end{lemma}

Assume that $A(\kappa,\omega)$ has a unique and simple eigenvalue
$\tilde\ell(\kappa,\omega)$ contained in a fixed disk centered at 0 in the complex $\lambda$-plane for
all $(\kappa,\omega)$ in a complex neighborhood of $\kwz\in{\cal D}\subset\mathbb{R}^2$.  (In fact, that $\kwz\in\mathbb{R}^2$ is not necessary for the present discussion.)
It will be convenient to work with $\ell=c\tilde\ell$ for a nonzero complex constant $c$ to be specified later.

Given a source field $\phi(\kappa,\omega)$ that is analytic at $\kwz$, consider the scattering problem
\begin{equation}\label{Ascattering}
  A(\kappa,\omega)\psi=\ell\phi.
\end{equation}
The analyticity of the field $\psi$ is proved in \cite[\S5.2]{Shipman2010}.  It analytically connects scattering states for $\ell\kw\not=0$ to generalized guided modes on the dispersion relation $\ell\kw=0$ near $\kwz$.

\begin{theorem}\label{theorem:AnaInverse}
The simple eigenvalue $\tilde\ell$ is analytic at $(\kappa_0,\omega_0)$, and, for any source field $\phi$ that is analytic at $(\kappa_0,\omega_0)$,
the solution $\psi(\kappa,\omega)$ is analytic at $(\kappa_0,\omega_0)$.
\end{theorem}

The analytic connection between scattering states and guided modes, introduced in \cite{ShipmanVenakides2005}, is achieved as follows.  One analytically resolves the identity operator on $\Honeper$ through the Riesz projections,
\begin{eqnarray}
  && I = P_1+P_2,\label{resolution1} \\
  && P_1(\kappa,\omega)=\frac{1}{2\pi i}\oint_C(\lambda I-A(\kappa,\omega))^{-1}d\lambda,\label{resolution2}
\end{eqnarray}
where $C$ is a sufficiently small positively oriented circle centered at 0 in the complex $\lambda$-plane.  The image of $P_1$ is the one-dimensional eigenspace of $A(\kappa,\omega)$ corresponding to the eigenvalue $\tilde \ell(\kappa,\omega)$ if this eigenvalue lies within the circle $C$.  This eigenspace is spanned by the analytic eigenvector
\begin{equation*}
  \hat\psi\kw = P_1\kw\hat \psi\kwz,
\end{equation*}
in which $\hat\psi\kwz$ is an eigenvector of $A\kwz$ corresponding to $\tilde\ell\kwz$.
The resolution (\ref{resolution1},\ref{resolution2}) provides an analytic decomposition of the source field $\phi$ near $\kwz$ as
$\phi=\alpha\hat\psi+\phi_2$, with
\begin{equation}\label{sourcedecomposition}
  \alpha\hat\psi = P_1\phi, \qquad \phi_2 = P_2\phi.
\end{equation}
Now, letting $A_2$ denote the restriction of $A$ to the range of $P_2$, one observes that the field
\begin{equation}\label{solutiondecomposition}
  \psi=c\alpha\hat\psi+\ell A_2^{-1} \phi_2,
\end{equation}
solves $A\psi=\ell\phi$\,:
\[
A\psi=\begin{bmatrix}
\tilde \ell & 0\\
0 & AP_2
\end{bmatrix}
\begin{bmatrix}
c\alpha\hat\psi \\
\ell A_2^{-1}\phi_2
\end{bmatrix}
=\begin{bmatrix}
\ell\alpha\hat\psi\\
\ell\phi_2
\end{bmatrix}
=\ell\phi.
\]
The Riesz projection naturally decomposes the source and solution fields into ``resonant" and ``nonresonant" parts via (\ref{sourcedecomposition},\ref{solutiondecomposition}).


\section{Resonant transmission for a symmetric slab} \label{Section:Single}

From now on, we will assume that the structure is symmetric about the $x$-axis.  Thus, in addition to the conditions \eqref{structure}, we assume also that $\epsilon(x,-z)=\epsilon(x,z)$ and $\mu(x,-z)=\mu(x,z)$ for all $x$.  We also assume that $\ell\kw$ is a simple (necessarily analytic) eigenvalue in a neighborhood of $\kwz\in{\cal D}$ and that $\ell\kwz=0$.

\subsection{The reduced scattering matrix} 

For $\kw\in{\cal D}$, consider the problem of scattering of the field $e^{i(\kappa x+\eta_0z)}$ incident upon the slab on the left.  By Theorem \ref{theorem:ExistenceSc}, a solution exists and the difference of any two solutions is evanescent; in fact, near $\kwz$ the solution is unique if and only if $\ell\kw\not=0$.  Thus the propagating components of the periodic part $u$ of the solution $\tilde u = e^{i\kappa x}u$ are unique, resulting in well-defined complex reflection and transmission coefficients $R$ and $T$,
\begin{eqnarray*}
  && u = e^{i\eta_0 z} + Re^{-i\eta_0 z} + \sum_{m\not=0} c_m^- e^{i(m x-\eta_m z)} \quad \text{for $z\leq -L$}, \\
  && u = Te^{i\eta_0 z} + \sum_{m\not=0} c_m^+ e^{i(m x+\eta_m z)} \quad \text{for $z\geq L$}.
\end{eqnarray*}
Because of the symmetry of the structure with respect to $z$, an incident field $e^{i(\kappa x-\eta_0z)}$ from the right produces identical reflection and transmission coefficients.  Thus the {\em reduced scattering matrix} for the structure for $\kw\in{\cal D}$ is
\begin{equation*}
  S\kw = \mat{1.2}{T\kw}{R\kw}{R\kw}{T\kw},
\end{equation*}
which gives the outward propagating components in terms of the inward propagating components in the expression (\ref{Fourier}) via
$S(A_0^-,B_0^+)^t=(A_0^+,B_0^-)^t$.

In terms of the transmission coefficient $T$, we define the {\itshape transmittance} to be the fraction of energy flux that is transmitted across the slab.  The transmittance is the quantitiy $|T|^2$, which lies in the interval $[0,1]$.

Let us now take the incident field from the left to be $\ell\kw e^{i(\kappa x+\eta_0 z)}$, which results in a reflected field $a\kw e^{i(\kappa x-\eta_0 z)}$ for $z\to-\infty$ and a transmitted field $b\kw e^{i(\kappa x+\eta_0 z)}$ for $z\to\infty$, with coefficients given by
\begin{equation}\label{ab}
  a=\ell R, \;\; b=\ell T.
\end{equation}
By the structural symmetry, an incident field $\ell\kw e^{i(\kappa x-\eta_0 z)}$ from the right results in a reflected field $a\kw e^{i(\kappa x+\eta_0 z)}$ for $z\to\infty$ and $b\kw e^{i(\kappa x-\eta_0 z)}$ for $z\to-\infty$, with coefficients also given by~\eqref{ab}.
The utility of working with $a$ and $b$ is that they are analytic, whereas $R$ and $T$ are not analytic at points $\kw$ where $\ell=0$\,.

\begin{lemma}  
The coefficients $a\kw$ and $b\kw$ are analytic in $\kappa$ and $\omega$.
\end{lemma}
\begin{proof}
The analyticity of the incident field $\ell\kw e^{i(\kappa x+\eta_0 z)}$ implies the analyticity of the source field $\ell\kw\phi\kw$ in the equation $A(\kappa,\omega)\psi=\ell\phi$ and hence, by Theorem \ref{theorem:AnaInverse}, also the analyticity of the solution field $\psi\kw=u(x,z;\kappa,\omega)$ in $\Honeper$.  The coefficients $a\kw$ and $b\kw$ of $\psi$ are given by
\begin{eqnarray*}
  && a\kw = \frac{e^{i\eta_0\kw L}}{2\pi} \int_0^{2\pi} u(x,-L;\kappa,\omega) dx, \\
  && b\kw = \frac{e^{-i\eta_0\kw L}}{2\pi} \int_0^{2\pi} u(x,L;\kappa,\omega) dx,
\end{eqnarray*}
and since $\eta_0\kw$ is analytic and $u\mapsto\int_0^{2\pi}u(x,\pm L)dx$ are bounded linear functionals on $\Honeper$, both $a$ and $b$ are analytic.
\end{proof}
This lemma provides a representation of the scattering matrix as the ratio of analytic functions in a complex neighborhood of $\cal D$, except at points of the dispersion relation $\ell\kw=0$,
\begin{equation}\label{ScatteringMatrix}
  S\kw = \frac{1}{\ell\kw} \mat{1.2}{b\kw}{a\kw}{a\kw}{b\kw}.
\end{equation}
Assuming that $\ell=0$ at an isolated point $\kwz$ of $\cal D$, we see that, in a real punctured neighborhood of $\kwz$, $S\kw$ is a complex-valued real-analytic function.  In fact, for real $\kw$, $S\kw$ is unitary, a standard fact that is shown by integration by parts:
\begin{equation} \label{eqn:S_unitary1}
\begin{cases}
|\ell|^2=|a|^2+|b|^2 ,\\
a\bar b+\bar a b=0.
\end{cases}
\end{equation}
This implies, in particular, that three analytic functions vanish at $\kwz$:
\begin{equation*}
  \ell\kwz = a\kwz = b\kwz = 0,
\end{equation*}
which is the feature that leads to the sensitive behavior of the transmission and reflection coefficients $|b/\ell|^2$ and $|a/\ell|^2$ near $\kwz$.

Now, each of $a\kw$ and $b\kw$ is a complex functions of two complex variables, and a 

In this section, we analyze the generic case
\begin{equation}\label{generic}
  \frac{\partial\ell}{\partial\omega}\ne0, \;\;
\frac{\partial a}{\partial\omega}\ne0, \;\;
\frac{\partial b}{\partial\omega}\ne0 \;\; \text{ at } \kwz.
\end{equation}
The Weierstra{\ss} Preparation Theorem tells us that the zero-sets of $a$, $b$, and $\ell$ are graphs of analytic functions of $\kappa$ near $\kappa_0$.
Let $\ka=\kappa-\kappa_0$ and $\om=\omega-\omega_0$.
With the appropriate choice of $c$ in $\ell=c\tilde\ell$, the Theorem yields the following factorizations:
\begin{equation} \label{eqn:single_expansion}
\begin{aligned}
a\kw&=(\om+r_1\ka+r_2\ka^2+\cdots)(r_0e^{i\gamma}+r_{\ka}\ka+r_{\om}\om+O(|\ka|^2+|\om|^2)),\\
b\kw&=(\om+t_1\ka+t_2\ka^2+\cdots)(it_0e^{i\gamma}+t_{\ka}\ka+t_{\om}\om+O(|\ka|^2+|\om|^2)),\\
\ell\kw&=(\om+\ell_1\ka+\ell_2\ka^2+\cdots)(1+\ell_{\ka}\ka+\ell_{\om}\om+O(|\ka|^2+|\om|^2)),
\end{aligned}
\end{equation}
in which $0<r_0<1$ and either $0<t_0<1$ or $-1<t_0<0$\,; the symbols $r_{\ka}$, $r_{\om}$, {\it etc.}, refer to constants.  One shows that the same unitary number $e^{i\gamma}$ appears in the second factors of both $a$ and $b$ by using the second expression in \eqref{eqn:S_unitary1}.  The zero-set of the first factor of each in each of these expressions coincides with the zero set of the corresponding function $a$, $b$, or $\ell$ near $\kwz$.


\smallskip
Under these conditions, one can deduce several properties of the coefficients; see \cite{ShipmanVenakides2005} and \cite[Theorem~10]{Shipman2010} for proofs.
\begin{lemma}\label{lemma:single}
The following relations hold among the coefficients in the form \eqref{eqn:single_expansion}:

\smallskip
\noindent
(i) $r_0^2+t_0^2=1$,\\
(ii) $\ell_1=r_1=t_1\in\RR$,\\
(iii) $\Im(\ell_2)\ge 0$,\\
(iv) $\ell_2\in\RR\iff r_2=t_2\in\RR\iff \ell_2=r_2=t_2\in\RR$.
\end{lemma}
When $\kappa_0=0$, the coefficient $\ell_1$ necessarily vanishes because of symmetry of the dispersion relation $\ell\kw=0$ in $\kappa$.
Whether $\ell_1\neq0$ can be realized for dielectric slabs remains an open problem.  Nevertheless, we do not assume $\ell_1=0$.
We also assume that $\Im(\ell_2)>0$, which is sufficient to guarantee that $\kwz$ is an isolated point of $\ell\kw=0$ in $\cal D$.

The proof of the following technical lemma is in the Appendix.
\begin{lemma}\label{prop:singlesym_alt}
Under the assumptions \eqref{generic} and $\Im(\ell_2)>0$,
one of the following alternatives is satisfied:
\smallskip

\noindent
(i) $r_2$ and $t_2$ are distinct real numbers; \\
(ii) $r_2=t_2\not\in\RR$ and either $\ell_2=r_2=t_2$ or $\ell_2=\bar r_2=\bar t_2$.
\end{lemma}

\subsection{Resonant transmission} \label{sec:generic} 

We now present and prove the main theorem on total reflection and transmission by symmetric periodic slabs.  The content of the three parts of Theorem \ref{Thm:Main} can be paraphrased as follows.

\medskip
(i)
If the coefficients $r_2$ and $t_2$ of the quadratic part of the expansions (\ref{azeroset},\ref{bzeroset}) are distinct numbers (case (i) of Lemma~\ref{prop:singlesym_alt}), then all of the coefficients $r_n$ and $t_n$  of both expansions turn out to be real numbers.  The consequence of this is that 
 $a\kw$ and $b\kw$ vanish along real-analytic curves in $\cal D$ given by
\begin{eqnarray}
  && \omega=\omega_0-\ell_1(\kappa-\kappa_0)-r_2(\kappa-\kappa_0)^2-\cdots, \quad(a\kw=0)
  \label{azeroset} \\
  && \omega=\omega_0-\ell_1(\kappa-\kappa_0)-t_2(\kappa-\kappa_0)^2-\cdots. \quad(b\kw=0)
  \label{bzeroset}
\end{eqnarray}
Because $r_2\not=t_2$, these curves are distinct and intersect each other tangentially at $\kwz$ with one of them remaining above the other.  They give the loci of $100\%$ transmission ($|b/\ell|=1$ when $a=0$) and $0\%$ transmission ($|b/\ell|=0$).  Thus we are presenting a proof of {\em total transmission and reflection at two nearby frequencies near $\omega_0$}, whose difference tends to zero as $\kappa$ tends to $\kappa_0$.  This establishes rigorously the numerically observed transmission spikes in \cite{ShipmanVenakides2005,Shipman2010}.  The zero sets of $a$ and $b$ are depicted in Fig.~\ref{fig:Diamond}.  The scattering matrix $S\kw$ is not continuous at $\kwz$ because $|T|=|b/\ell|$ takes on all values between $0$ and $1$ in each neighborhood of $\kwz$.  However, as a function of $\omega$ at $\kappa=\kappa_0$, $S(\kappa_0,\omega)$ is in fact continuous.

\smallskip
(ii)
The result of part (i) of Theorem~\ref{Thm:Main} provides a local representation of the zero sets of $a$ and $b$ about $\kwz$ as the graphs of the real analytic functions (\ref{azeroset},\ref{bzeroset}) in $\cal D$.  These locally defined functions can be extended to real-analytic functions of $\kappa$ in an interval around $\kappa_0$ up until their graphs either intersect the boundary of the diamond $\cal D$ or attain an infinite slope.  It is not known which of these possibilities actually occur in theory.

\smallskip
(iii)
In case (ii) of Lemma~\ref{prop:singlesym_alt}, we prove that the transmittance $|b\kw/\ell\kw|^2$ is continuous at $\kwz$.  We are aware of no examples of this case but are not able to rule it out theoretically.  This result tells us that, if there is in fact a resonant transmission spike, then case (i) of the Lemma must hold and thus this spike attains the extreme values $0$ and $1$ by part (i) of Theorem~\ref{Thm:Main}.  In short, either $|b\kw/\ell\kw|^2$ is continuous at $\kwz$ or it attains $0$ and $1$ in every neighborhood of $\kwz$.



\begin{theorem}[Total transmission and reflection]\label{Thm:Main}
Given a two-dimensional lossless periodic slab satisfying \eqref{structure} that is symmetric about a line parallel to it.  Let $\kwz$ be a wavenumber-frequency pair in the region $\cal D$ of exactly one propagating harmonic at which the slab admits a guided mode, that is $\ell\kwz=0$.  Assume in addition the generic condition (\ref{generic}) and that $\Im(\ell_2)>0$ in the expansion of $\ell$ in \eqref{eqn:single_expansion}.  Then either the transmittance is continuous at $\kwz$ or it attains the magnitudes of \,$0$ and $1$ on two distinct real-analytic curves that intersect quadratically at $\kwz$.  Specifically,

\smallskip
(i)
If $r_2\not=t_2$, then $r_n$ and $t_n$ are real for all $n$ and both $a(\kappa_0,\omega)/\ell(\kappa_0,\omega)$ and \linebreak $b(\kappa_0,\omega)/\ell(\kappa_0,\omega)$ can be extended to continuous functions of $\omega$ in a real neighborhood of $\omega_0$ with values $r_0e^{i\gamma}$ and $it_0e^{i\gamma}$ at $\omega_0$.

(ii)
If $r_2\not=t_2$, let $c$ denote either $a$ or $b$ and let $\omega=f(\kappa)$ denote the corresponding function from the pair (\ref{azeroset},\ref{bzeroset}).  Then $f$ can be extended to a real analytic function $g(\kappa)$ on an interval $(\kappa_1,\kappa_2)$ containing $\kappa_0$ such that the graph of $g$ is in $\cal D$ and for each $i=1,2$, the limit $g(\kappa_i):=\lim_{\kappa\rightarrow\kappa_i} g(\kappa)$ exists and either $(\kappa_i,g(\kappa_i))$ is on the boundary of $\cal D$ or
$\frac{\partial c}{\partial\omega}(\kappa_i,g(\kappa_i))=0$.

(iii)
If $r_2=t_2$, then $|a/\ell|$ and $|b/\ell|$ can be extended to continuous functions in a real neighborhood of $\kwz$ with values $|r_0|$ and $|t_0|$ at $\kwz$.
\end{theorem}
\begin{proof}  

(i) Assume $r_2\ne t_2\in\mathbb{R}$.
Recall from Lemma~\ref{lemma:single} that $r_1=t_1=\ell_1\in\RR$.
Assuming $r_1,\ldots,r_n\in\mathbb{R}$ for $n\geq2$, we will show that $r_{n+1}\in\mathbb{R}$.
For $(\ka,\om)$ subject to the relation $\om+\ell_1\ka+r_2\ka^2+r_3\ka^3+\ldots+r_n\ka^n=0$,
\begin{equation*}
\begin{split}
\frac{a}{b}&=\frac{r_{n+1}\ka^{n+1}+r_{n+2}\ka^{n+2}+\cdots}{(t_2-r_2)\ka^2+(t_3-r_3)\ka^3
   +\cdots+(t_n-r_n)\ka^n+t_{n+1}\ka^{n+1}+\cdots}\left[\frac{r_0e^{i\gamma}+O(|\ka|)}{i t_0e^{i\gamma}+O(|\ka|)}\right]\\
&=\ka^{(n-1)}\frac{r_{n+1}+r_{n+2}\ka+\cdots}{(t_2-r_2)+(t_3-r_3)\ka+\cdots}
\left[\frac{r_0e^{i\gamma}+O(|\ka|)}{i t_0e^{i\gamma}+O(|\ka|)}\right].
\end{split}
\end{equation*}
Because $a\bar b+\bar ab=0$ (see \eqref{eqn:S_unitary1}), it follows that $a/b\in i\mathbb{R}$, and thus
\begin{equation*}
\left[\frac{r_{n+1}+r_{n+2}\ka+\cdots}{(t_2-r_2)+(t_3-r_3)\ka+\cdots}\right]
  \left[\frac{r_0e^{i\gamma}+O(\ka)}{i t_0e^{i\gamma}+O(\ka)}\right]\in i\mathbb{R}.
\end{equation*}
Letting $\ka\rightarrow0$ yields
\begin{equation*}
\frac{r_{n+1}}{t_2-r_2}\cdot \frac{r_0}{it_0}\in i\mathbb{R},
\end{equation*}
which implies that $r_{n+1}\in\mathbb{R}$.  We conclude by induction that $r_n\in\mathbb{R}$ for all $n\geq1$.  The proof that $t_n\in\mathbb{R}$ for all $n$ is analogous.

To prove the second statement, one sets $\ka=\kappa-\kappa_0=0$ and observes that the ratios $a/\ell$ and $b/\ell$ have limiting values of $r_0e^{i\gamma}$ and $it_0e^{i\gamma}$, respectively, as $\om\to0$, or $\omega\to\omega_0$.

(ii)
Define the set
\begin{equation*}
  {\cal G} := \{ g:(\kappa_-,\kappa_+)\!\to\!\mathbb{R} \,|\, \kappa_0\in(\kappa_-,\kappa_+),\, g \text{ is real-analytic},\,
             a(\kappa,g(\kappa))=0,\, \Gamma(g)\in{\cal D}\},
\end{equation*}
in which $\Gamma(g)$ is the graph of $g$, and define the numbers
\begin{eqnarray*}
  && \kappa_1 := \inf\{ \kappa_- \,|\, g\in{\cal G}\}, \\
  && \kappa_2 := \sup\{ \kappa_+ \,|\, g\in{\cal G}\}.
\end{eqnarray*}
By virtue of the function \eqref{azeroset}, which belongs to $\cal G$, $\kappa_1<\kappa_0<\kappa_2$.  Standard arguments show that any two functions from $\cal G$ coincide on the intersection of their domains, and one obtains thereby a maximal extension $g\in{\cal G}$ of \eqref{azeroset} with domain $(\kappa_1,\kappa_2)$.  We now show that $\lim_{\kappa\nearrow\kappa_2}g(\kappa)$ exists.  Set
$\omega_-=\liminf_{\kappa\nearrow\kappa_2} g(\kappa)$ and $\omega_+=\limsup_{\kappa\nearrow\kappa_2}g(\kappa)$.  Because of the continuity of $g$, the segment $(\kappa_2,[\omega_-,\omega_+])$ in $\overline{\cal D}$ is in the closure of the graph of $g$, on which $a$ vanishes.  Thus $a(\kappa_2,\omega)=0\,\forall\,\omega\in[\omega_-,\omega_+]$.  Moreover, for each $\omega\in(\omega_-,\omega_+)$, there is a sequence of points $(\kappa_j,\omega)$ with $\kappa_j\nearrow\kappa_2$ and $a(\kappa_j,\omega)=0$ from which we infer that
$\partial^n a/\partial\kappa^n(\kappa_2,\omega)=0\,\forall\,n\in\mathbb{N}$ and hence that
$\partial^{m+n}a/\partial\omega^m\partial\kappa^n(\kappa_2,\omega)=0\;\forall\,m,n\in\mathbb{N}$ and $\forall\,\omega\in(\omega_-,\omega_+)$.  If $(\omega_-,\omega_+)$ is nonempty, then $a$ must vanish in $\cal D$, which is untenable in view of the assumption that
$\partial a/\partial\omega\kwz\not=0$.  This proves that $\omega_-\!=\!\omega_+$ so that $\omega_2:=\lim_{\kappa\nearrow\kappa_2}g(\kappa)$ exists.  If $|\kappa|/\sqrt{\ep_0\mu_0}<\omega_2<(1-|\kappa|)/\sqrt{\ep_0\mu_0}$ and $\partial a/\partial\omega(\kappa_2,\omega_2)\not=0$, the implicit function theorem provides an element of $\cal G$ with $\kappa_+>\kappa_2$, which is not compatible with the definition of $\kappa_2$.
Analogous arguments apply to the endpoint $\kappa_1$ and to the function $b$.

(iii) If $r_2=t_2$, then by Lemma \ref{prop:singlesym_alt}, $t_2=\ell_2$ or $t_2=\bar\ell_2$.  Keeping in mind that $\Im(\ell_2)>0$ and $\ell_1\in\mathbb{R}$ and restricting to $(\ka,\om)\in\mathbb{R}^2$,
\[
\begin{split}
\lim_{\ka,\om\to0} \left|\frac{b}{\ell}\right|
 =&\lim_{\ka,\om\to0}
 \left|\frac{\om+\ell_1\ka+t_2\ka^2+t_3\ka^3+\cdots}
 {\om+\ell_1\ka+\ell_2\ka^2+\ell_3\ka^3+\cdots}\right|
 \left|\frac{it_0e^{i\gamma}+O(|\om|+|\ka|)}{1+O(|\om|+|\ka|)}\right| \\
 =&\lim_{\ka,\om\to0}|t_0|
 \left|\frac{\om+\ell_1\ka+t_2\ka^2}
 {\om+\ell_1\ka+\ell_2\ka^2}\right|.
\end{split}
\]
Whether $t_2$ is equal to $\ell_2$ or $\bar\ell_2$, the second factor of the last expression is equal to unity, and we obtain $\lim_{\ka,\om\to0} |b/\ell| = |t_0|$.  Similarly, one shows that $\lim_{\ka,\om\to0} |a/\ell| = |r_0|$.
\end{proof}


\subsection{Discussion of assumptions and conditions}  

A number of assumptions and conditions concerning the slab structure, the nature of the dispersion relation $\ell\kw=0$, and the coefficients in the Weierstra{\ss} representations of $\ell\kw$, $a\kw$, and $b\kw$ have played roles in our analysis.  We take a moment to discuss their relevance and the regimes of their validity.

\smallskip
{\itshape Existence of embedded guided-mode frequencies.}\, Real $\kw$ values at which a periodic slab supports a guided mode are typically located below the light cone $\omega^2\epsilon_0\mu_0=\kappa^2$ for the ambient medium in the first Brillouin zone, that is below the diamond $\cal D$.  The reason is that the light cone marks the bottom of the continuous spectrum for the $\kappa$-periodic Helmholtz operator in the strip $\{(x,z):|x|<\pi\}$ \cite{Bonnet-BeStarling1994} and thus any guided mode frequency above the light cone corresponds to an embedded eigenvalue and is therefore unstable with respect to generic perturbations of the operator.  Embedded guided-mode frequencies arise under special conditions; in particular, structures symmetric in $x$ for which $\epsilon$ and $\mu$ are sufficiently large admit a positive finite number of such frequencies in $\cal D$ at $\kappa=0$ \cite{Bonnet-BeStarling1994}, which are antisymmetric in $x$.

{\itshape Symmetry.}\, The assumption of symmetry of the structure about a line parallel to it is crucial to the proof of total reflection and transmission.  The condition is used in part (i) of the proof of Theorem~\ref{Thm:Main} when the fact that $a/b$ is imaginary is invoked.  This is a technical step in the proof and does not provide much of an enlightening explanation of the role of symmetry.  At least one can say that the functions (\ref{azeroset},\ref{bzeroset}), which are in general complex-valued, become real-valued when the structure is subject to a certain condition of symmetry.

{\itshape Simplicity of the eigenvalue $\ell$.}\, The frequencies of these antisymmetric modes are eigenvalues of a modified scattering problem in which the one imaginary multiplier $-i\eta_0=-i|\eta_0|$ in the definition of the Dirichlet-to-Neumann operator $\DtN$ \eqref{DtN} is replaced by zero.
This results in a symmetric modification $h_r(u,v)$ of the non-symmetric form $h(u,v)$ defined in~\eqref{a}. For $\kappa=0$, the associated operator is invariant on the subspace of $\Honeper$ consisting of antisymmetric functions in~$x$.  Thus it can be viewed as a Dirichlet operator in the strip $\{(x,z):0<x<\pi\}$, and the smallest eigenvalue is therefore simple.

{\itshape The generic condition \eqref{generic}} prevails in numerical computations.  The nonvanishing of $\partial\ell/\partial\omega$ actually implies the nonvanishing of either $\partial a/\partial\omega$ or $\partial b/\partial\omega$.  If one of the latter vanishes, the limiting value of the transmitted energy at $\kwz$ is either $1$ or $0$; we will consider this case in the next section.  Verification of any of these conditions for a particular structure is done numerically (or by experiment).

{\itshape Nonrobustness.}\, Guided modes at embedded frequencies tend to be nonrobust under perturbations of $\kappa$, which means that the point $\kwz$ is isolated in the domain $\cal D$.  Typically, one computes that $\Im(\ell_2)$ is strictly positive in \eqref{eqn:single_expansion}, which guarantees that $\kwz$ is an isolated point in the intersection of $\cal D$ with the dispersion curve $\ell\kw=0$.  So far, we do not have an analytic proof of nonrobustness for any particular structure.  In fact, one can show that real dispersion relations for robust guided modes do exist within the continuous spectrum if the structure possesses a smaller period than that of the pseudo-period of the guided modes.

{\itshape Traveling modes and $\ell_1\not=0$.}\,  A guided mode at $\kappa=0$ is in fact a standing wave, and because of the symmetry of $\ell\kw$ in $\kappa$, necessarily $\ell_1=0$.  In this case, the anomaly is not linearly detuned as $\kappa$ is perturbed from $0$ but simply widens at most quadratically in $\kappa$.  As far as we know, guided modes that are not robust under a perturbation of $\kappa$ have not been demonstrated for $\kappa_0\not=0$ or for structures that are not symmetric in $x$, although we believe that they should exist in slabs with anisotropic components.  Guided modes with $\kappa_0\not=0$ have been proved in a discrete system in which the ambient space is a two-dimensional lattice and the waveguide is a one-dimensional periodic lattice coupled along a line to the 2D one \cite{PtitsynaShipman2011} and for asymmetric scatterers in another discrete model \cite{ShipmanRibbeckSmith2010}.  In fact, both of these systems are amenable to the analysis in this section and are subject to Theorem \ref{Thm:Main} with $\ell_1$ not necessarily zero.  Compare Figures \ref{fig:SingleRealTwo} and \ref{fig:SingleRealOne}.  

{\itshape The alternative $r_2=t_2\notin\mathbb{R}$} of Lemma \ref{prop:singlesym_alt}.
Under this condition, Theorem \ref{Thm:Main} states that the transmittance exhibits no spike; in fact, $b/\ell$ is continuous at $\kwz$.  Although we are not able to rule out this alternative, we are aware of no examples in which it does occur numerically.  In any case, Theorem \ref{Thm:Main} guarantees that, if the transmittance is discontinuous at $\kwz$, it must attain the values of 0 and 1 along the real-analytic curves (\ref{azeroset},\ref{bzeroset}).


\section{Nongeneric resonant transmission}\label{Section:Multiple}

The anomalies we have analyzed so far possess ``background" reflection and transmission values of $|r_0|$ and $|t_0|$, both of which lie strictly between $0$ and $1$.  These are the limiting values of the ratios of reflected or transmitted energy to the incident energy at $\kappa=\kappa_0$ and $\omega\to\omega_0$ (see part (i) of Theorem \ref{Thm:Main}).  When one of these values reaches $0$ or $1$ (we always have $|r_0|^2 + |t_0|^2=1$), the form for $a$ or $b$ in  \eqref{eqn:single_expansion} must be modified so that the order of $\om$ in its first factor exceeds one.
This results in a transmission anomaly possessing a single sharp dip or peak near the frequency of the guided mode---the peak-dip feature reduces to an upright or inverted Lorentzian shape (Fig.~\ref{fig:MixRealOne1}, \ref{fig:MixRealThree1}).

We also discuss what happens when $\ell$ is degenerate at $\kwz$.  We still assume that $\ell$ is a simple eigenvalue of the operator $A\kw$ but allow the first partial derivative
of $\ell$ with respect to $\omega$ to vanish at $\kwz$.  Assuming that the second partial is nonzero, two double-spiked anomalies instead of one emanate from the guided mode frequency upon perturbation of~$\kappa$ (Fig.~\ref{fig:SecondOrderRealOne1}, \ref{fig:SecondOrderRealTwo1}).  We do not have a rigorous demonstration of a case of degeneracy, but we have observed numerically a pair of anomalies emerging from a single guided-mode frequency for a slab consisting of two infinite rows of identical dielectric rods side by side.

We shall omit the proofs, which are technical.  Resonant transmission in nongeneric cases is worthy of an in-depth investigation in its own right, and our aim here is to provide a view of some of the phenomena that can theoretically occur.

\subsection{Total background reflection or transmission}\label{sec:totalbackground}

Let us assume that $\partial\ell/\partial\omega\not=0$ at $\kwz$; then only one of the functions $a$ and $b$ can be degenerate at $\kwz$.

\begin{proposition}
Suppose that at $\kwz\in{\cal D}$,
\begin{equation*}
  \ell\kwz=0,\;\text{ and }\;\; \frac{\partial\ell}{\partial\omega}\kwz\not=0.
\end{equation*}
Then
\begin{equation*}
  \frac{\partial a}{\partial\omega}\kwz\not=0 \;\,\text{ or }\;\;
  \frac{\partial b}{\partial\omega}\kwz\not=0.
\end{equation*}
\end{proposition}
\begin{proof}
  By conservation of energy $|\ell|^2 = |a|^2+|b|^2$ for $\kw\in{\cal D}$, we have $a\kwz=0$ and $b\kwz=0$ and therefore the representations $\ell(\kappa_0,\omega) = (\omega-\omega_0) g_1(\omega)$,
$a(0,\omega) = (\omega-\omega_0)^m g_2(\omega)$, and
$b(0,\omega) = (\omega-\omega_0)^n g_3(\omega)$, with $g_i$ analytic and nonzero at $\omega_0$.  This is consistent with $|\ell|^2 = |a|^2+|b|^2$ only if $m=1$, that is ${\partial a}/{\partial\omega}\not=0$, or $n=1$, that is ${\partial b}/{\partial\omega}\not=0$, at $\kwz$.
\end{proof}

Let us consider the case of 100\% background transmission, that is ${\partial\ell}/{\partial\omega}$,
${\partial^2 a}/{\partial\omega^2}$, and
${\partial b}/{\partial\omega}$ are nonzero and that
${\partial a}/{\partial\omega}=0$ at $(\kappa_0,\omega_0)$.
The Weierstra{\ss} Preparation Theorem yields the following factorizations:
\begin{eqnarray*}\label{formula:mixed_basic}
&& \ell\kw=\left(\om+\ell_1\ka+\ell_2\ka^2 +\cdots\right)(1+O(|\ka|+|\om|)), \\
&& a\kw=\left(\om^2+\om\alpha^1(\ka)+\alpha^0(\ka)\right)(r_0e^{i\alpha}+O(|\ka|+|\om|)), \\
&& b\kw=\left(\om+t_1\ka+t_2\ka^2+\cdots\right)(t_0e^{i\beta}+O(|\ka|+|\om|)),
\end{eqnarray*}
where $r_0,t_0>0$.

Let us also suppose that $\om^2+\om\alpha^1(\ka)+\alpha^0(\ka)$ has distinct roots at $\ka=0$, so that it can be factored analytically,
\[
a=\left(\om+r_1^{(1)}\ka+r_2^{(1)}\ka^2+\cdots\right)
  \left(\om+r_1^{(2)}\ka+r_2^{(2)}\ka^2+\cdots\right)
  \left(r_0e^{i\alpha}+O(|\ka|+|\om|)\right).
\]
Assuming that the coefficients $r_n^{(i)}$ are real, maximal transmission is achieved now at {\em two} points that move apart linearly in $\ka=\kappa-\kappa_0$.

\begin{lemma} \label{Prop:MixedBacic}
The coefficients of the expansions satisfy the following properties. \\
(i) $t_0=1$, $t_1=\ell_1\in\mathbb{R}$, $\Im(\ell_2)\ge0$; \\
(ii) $e^{i\beta}= \pm ie^{i\alpha}$;\\ 
(iii) $r_1^{(1)}+r_1^{(2)}$,  $r_1^{(1)}r_1^{(2)}$,
and $(r_1^{(1)}-\ell_1)(r_1^{(2)}-\ell_1)$ are real-valued.
Therefore, either $r_1^{(1)},r_1^{(2)}$ are both real or they are conjugate complex numbers.

\smallskip
One of the following alternatives holds. \\ 
(i) If $(r_1^{(1)}-\ell_1)(r_1^{(2)}-\ell_1)\ne0$, then $t_2=\Re(\ell_2)$ and
$|\Im(\ell_2)|^2=r_0^2\left|r_1^{(1)}-\ell_1\right|^2 \cdot\left|r_1^{(2)}-\ell_1\right|^2\ne0$,\\
(ii) If $(r_1^{(1)}-\ell_1)(r_1^{(2)}-\ell_1)=0$, then
$\Re(\ell_2)=\Re(t_2)$ and
$|\Im(\ell_2)|=|\Im(t_2)|$.
\end{lemma}

\smallskip
\noindent
We obtain the factorizations
\begin{equation}\label{eqn:mixed_expansion}
\begin{aligned}
\ell&=\left(\om+\ell_1\ka+\ell_2\ka^2
     +\cdots\right)(1+O(|\ka|+|\om|)),\\
a&=\left(\om+r_1^{(1)}\ka+r_2^{(1)}\ka^2+\cdots\right)
  \left(\om+r_1^{(2)}\ka+r_2^{(2)}\ka^2+\cdots\right)
  \left(r_0e^{i\gamma}+O(|\ka|+|\om|)\right),\\
b&=\left(\om+\ell_1\ka+t_2\ka^2+\cdots\right)
  \left(\pm ie^{i\gamma}+O(|\ka|+|\om|)\right).
\end{aligned}
\end{equation}

If all the coefficients $t_n$ are real-valued, which is guaranteed by the first alternative (see Theorem \ref{theorem:MixDipPeak}), then the transmission vanishes along the real-analytic curve
\begin{equation}\label{eqn:mixcurve1}
\omega=\omega_0-\ell_1(\kappa-\kappa_0)-t_2(\kappa-\kappa_0)^2-\cdots.  \quad(b=0)
\end{equation}
If $r_n^{(1)}$ and $r_n^{(2)}$ are real-valued then the transmission achieves $100\%$ along two real-analytic curves
\begin{equation}\label{eqn:mixcurve2}
\omega=\omega_0-r_1^{(i)}(\kappa-\kappa_0)-r_2^{(i)}(\kappa-\kappa_0)^2-\dots,  \; i=1, 2.  \quad(a=0)
\end{equation}
The two frequencies of total transmission and the one frequency of total reflection come together at $\omega_0$ as $\kappa$ tends to $\kappa_0$.  The frequencies of maximal transmission are difficult to detect in Fig.~\ref{fig:MixRealOne1} and \ref{fig:MixRealThree1} because the transmission function is nearly constant near these points.  A magnified view is shown on the right of Fig.~\ref{fig:MixRealThree1}.
The proof of the following theorem is similar to that of Theorem~\ref{Thm:Main}.  We believe that $t_1$ must lie between $r_1^{(1)}$ and $r_1^{(2)}$ but do not have a proof of this.

\begin{theorem}\label{theorem:MixDipPeak}

Suppose that $\epsilon$ and $\mu$ are symmetric in $z$ and that $\ell\kwz=0$ at $\kwz\in{\cal D}$.  Let $\frac{\partial\ell}{\partial\omega}$, $\frac{\partial^2 a}{\partial\omega^2}$, and $\frac{\partial b}{\partial\omega}$ be nonzero and
$\frac{\partial a}{\partial\omega}=0$ at $(\kappa_0,\omega_0)$.  Suppose in addition that $(r_1^{(1)}-\ell_1)(r_1^{(2)}-\ell_1)\ne0$.  Then

(i) $t_n$ is real for all $n$ and therefore the coefficient $b$ of the transmitted field vanishes on the real-analytic curve in a neighborhood of $(\kappa_0,\omega_0)$ given by (\ref{eqn:mixcurve1});

(ii) if $r_1^{(1)}$ and $r_1^{(2)}$ are distinct real numbers, then
$r_n^{(1)}$ and $r_n^{(2)} $ are real for all $n$ and therefore the coefficient $a$ of the reflected field vanishes on the real analytic curves given by (\ref{eqn:mixcurve2}).
\end{theorem}
%


\subsection{Multiple anomalies}\label{sec:multiple}

If both $\ell$ and $\partial\ell/\partial\omega$ vanish at $\kwz\in{\cal D}$, then again by virtue of the equation $|\ell|^2=|a|^2+|b|^2$ for real $\kw$, both $a$ and $b$ and their $\omega$-derivatives vanish at $\kwz$.  We treat the case in which
\begin{equation}\label{degenerate}
\begin{gathered}
\frac{\partial \ell}{\partial\omega}=0,\;
\frac{\partial a}{\partial\omega}=0,\;
\frac{\partial b}{\partial\omega}=0\,;\\
\frac{\partial^2\ell}{\partial\omega^2}\neq0,\;
\frac{\partial^2 a}{\partial\omega^2}\neq0,\;
\frac{\partial^2 b}{\partial\omega^2}\neq0\,.
\end{gathered}
\end{equation}
The zero loci are given locally by the roots of a quadratic function in $\om$ with coefficients that are analytic in $\ka$ and vanish at $\ka=0$,
\begin{eqnarray*}
&& \ell\kw=\left(\om^2+\om\lambda^1(\ka)+\lambda^0(\ka)\right)(1+O(|\ka|+|\om|)),\\
&& a\kw=\left(\om^2+\om\alpha^1(\ka)+\alpha^0(\ka)\right)(r_0e^{i\gamma}+O(|\ka|+|\om|)),\\
&& b\kw=\left(\om^2+\om\beta^1(\ka)+\beta^0(\ka)\right)(it_0e^{i\gamma}+O(|\ka|+|\om|)),
\end{eqnarray*}
in which $\lambda^i(0)=\alpha^i(0)=\beta^i(0)=0$, $0<r_0<1$, and $t_0$ is real with $0<|t_0|<1$.  The ratio of the leading coefficients in the factors of $a$ and $b$ is imaginary because of the special form of the scattering matrix (\ref{ScatteringMatrix}) coming from the symmetry of the structure in $z$.
Let us assume again that the left-hand factors have distinct roots so that they can be factored analytically.  Here again, we are able to show that just two real coefficients $\ell_1^{(1)}$ and $\ell_1^{(2)}$ for the linear terms $\ka$ suffice, as long as they are distinct,
\begin{equation}
\begin{aligned}\label{expansions3}
&\ell\kw=\left(\om+\ell_1^{(1)}\ka+\ell_2^{(1)}\ka^2
     +\cdots\right)\left(\om+\ell_1^{(2)}\ka+\ell_2^{(2)}\ka^2+\cdots\right)(1+O(|\ka|+|\om|)),\\
&a\kw=\left(\om+\ell_1^{(1)}\ka+r_2^{(1)}\ka^2+\cdots\right)
  \left(\om+\ell_1^{(2)}\ka+r_2^{(2)}\ka^2+\cdots\right)
  \left(r_0e^{i\gamma}+r_{\ka}\ka+r_{\om}\om+\cdots\right),\\
&b\kw=\left(\om+\ell_1^{(1)}\ka+t_2^{(1)}\ka^2+\cdots\right)
  \left(\om+\ell_1^{(2)}\ka+t_2^{(2)}\ka^2+\cdots\right)
  \left(it_0e^{i\gamma}+t_{\ka}\ka+t_{\om}\om+\cdots\right).
\end{aligned}
\end{equation}

\begin{lemma}\label{prop:SecondBasicProperty}
The coefficients $\ell_1^{(1)}$ and $\ell_1^{(2)}$ are real and the coefficients $\ell_2^{(1)}$ and $\ell_2^{(2)}$ have nonnegative imaginary parts.
Moreover, if $\Im(\ell_2^{(1)})>0$ and $\Im(\ell_2^{(2)})>0$,
then, for each $i\in\{1,2\}$, either $r_2^{(i)}$~and~$t_2^{(i)}$ are distinct real numbers or they are equal and not real.
\end{lemma}

Subject to structural symmetry with respect to $z$, one can prove a theorem analogous to those for the previous cases.  It says in particular that, if the first alternative in the last statement of the lemma holds for $i=1$ and $i=2$, then two peak-dip anomalies emerge from $\omega_0$ as $\kappa$ is perturbed from $\kappa_0$, each of which attains the values 0 and 1 along real-analytic curves passing through $\kwz$.

\begin{theorem}
Suppose that $\Im(\ell_2^{(i)})>0$ for all $i\in\{1,2\}$.\\
(i) If $r_2^{(i)}\not=t_2^{(i)}$ with $i\in\{1,2\}$, then $r_n^{(i)}$ and $t_n^{(i)}$ are real for all $n$.\\
(ii) If $r_2^{(i)}=t_2^{(i)}$ for all $i\in\{1,2\}$, then $a/\ell$ and $b/\ell$ are continuous at $\kwz$.
\end{theorem}

\section{Transmission graphs}\label{Section:Graphs}

We demonstrate resonant transmission anomalies with various choices of the coefficients in the expansions (\ref{eqn:single_expansion},\ref{eqn:mixed_expansion},\ref{expansions3}) of $\ell$, $a$, and $b$.  We graph the transmittance
\begin{equation}
|T(\kappa,\omega)|^2=\left|\frac{b(\kappa,\omega)}{\ell(\kappa,\omega)}\right|^2 =\frac{|b|^2}{|a|^2+|b|^2},
\end{equation}
retaining terms up to quadratic order in $\ka=\kappa-\kappa_0$ in the factors that vanish at $\kwz$ and only the constant in the nonzero factor.
 An error estimate of ${\cal O}(|\ka|+\om^2)$ for these approximations in the generic case is proved in \cite[Thm.~16]{PtitsynaShipman2011}.
First we demonstrate the nondegenerate case of Section \ref{sec:generic}, in which $r_2$ and $t_2$ are distinct real numbers (Fig.~\ref{fig:SingleRealTwo}, \ref{fig:SingleRealOne}).  For $\kappa=\kappa_0$ ($\ka=0$) the anomaly is absent; it widens quadratically in $\ka$, namely as $|t_2\!-\!r_2|\ka^2$.  If $\ell_1\not=0$, as in Fig.~\ref{fig:SingleRealOne}, then the anomaly is detuned from $\omega_0$ at a linear rate in $\ka$, whereas it widens only quadratically.

Figures \ref{fig:MixRealOne1} and \ref{fig:MixRealThree1} show the degenerate case in which the anomaly is a single dip descending to zero from a background of full transmission or a single peak rising to unity from a background of no transmission (Section~\ref{sec:totalbackground}).  In the former case, for example, full transmission is actually achieved at precisely two frequencies near $\omega=\omega_0$ ($\om=0$); this can be seen clearly in the magnified, right-hand image of Fig.~\ref{fig:MixRealThree1}.

When the first derivative with respect to $\omega$ of all three functions $\ell$, $a$, and $b$ vanishes at $\kwz$, an anomaly characterized by a pair of peak-dip features can emanate from $\omega_0$, as discussed in Section~\ref{sec:multiple}.  Figures \ref{fig:SecondOrderRealOne1}  and \ref{fig:SecondOrderRealTwo1} illustrate two possible choices of constants in the Weierstra{\ss} expansions.  At this point there are no rigorous results on which sets of constants are admissible or prohibited by periodic slabs or any other open waveguide system.

\smallskip
The vertical line in all graphs shows the location of $\om=0$, or $\omega=\omega_0$.

\begin{figure}
\centering
    \includegraphics[scale=0.9]{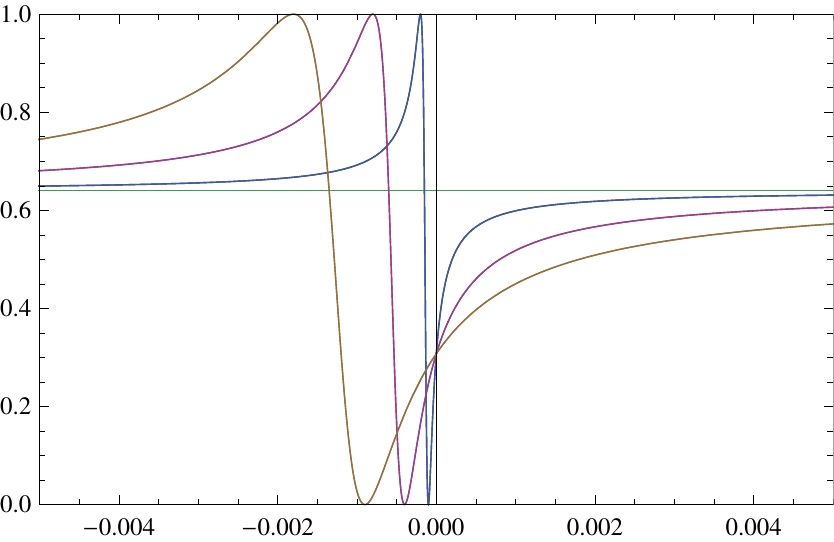}\quad
    \includegraphics[scale=0.9]{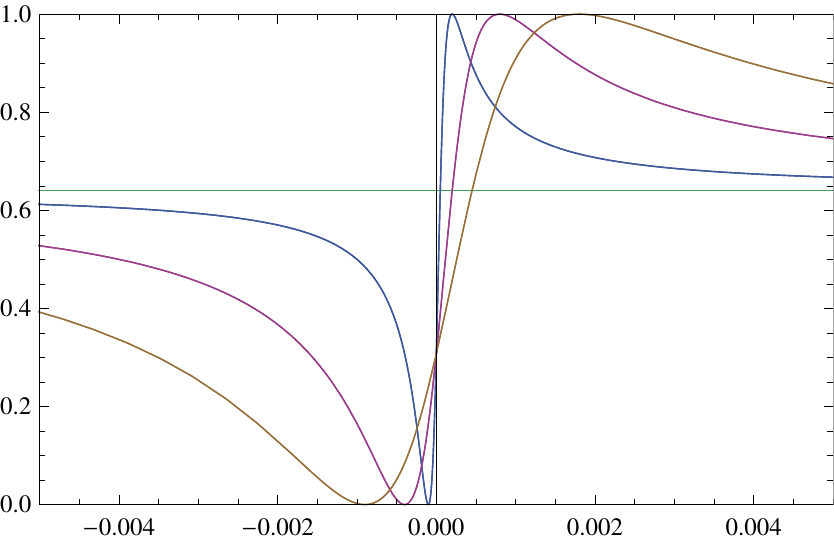}
    \caption{ \small
      $|T|^2$ as a function of $\om$ for $\ka=0,\pm0.01,\pm0.02,\pm0.03$.  The generic conditions $\frac{\partial\ell}{\partial\omega}, \frac{\partial a}{\partial\omega},\frac{\partial b}{\partial\omega}\ne0$ are satisfied at the bound-state pair $(\kappa_0,\omega_0)$.  In (\ref{azeroset},\ref{bzeroset}), $\ell_1=0$ so that there is no linear detuning of the anomaly with $\ka$.  Left: $0<t_2=1<r_2=2$ so that the peak is to the left of the dip and both are to the left of $\omega_0$. \,Right: $r_2=-2<0<t_2=1$.  In both graphs, $r_0=0.6$, $t_0=0.8$.  The transmission is symmetric in~$\ka$, and the curve without an anomaly (the horizontal curve) is the transmission graph for $\ka=0$.}
    \label{fig:SingleRealTwo}
\end{figure}
\begin{figure}
\centering
    \includegraphics[scale=0.9]{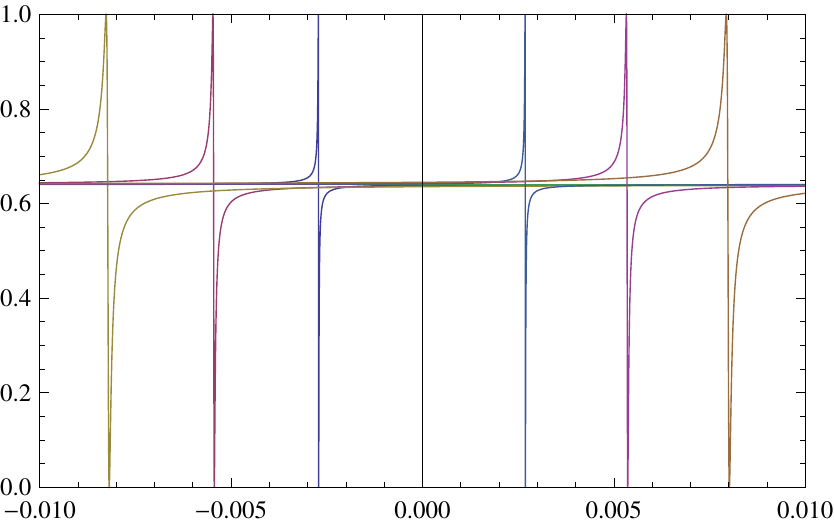}
    \caption{ \small
      $|T|^2$ as a function of $\om$ for $\ka=0,\pm0.003,\pm0.006,\pm0.009$.
      The generic conditions $\frac{\partial\ell}{\partial\omega},
       \frac{\partial a}{\partial\omega},\frac{\partial b}{\partial\omega}\ne0$ are satisfied at the bound-state pair $(\kappa_0,\omega_0)$.
       In (\ref{azeroset},\ref{bzeroset}), $\ell_1=0.9\not=0$, so the anomaly is detuned from $\omega=\omega_0$ ($\om=0$) in a linear manner in $\ka$.
       The coefficients $r_2=2$ and $t_2=1$ of $\ka^2$ are distinct real numbers, and $(r_0,t_0)=(0.6,0.8)$.
       The graph for $\ka=0$ is horizontal at $|T|^2\approx0.64$.
     }
    \label{fig:SingleRealOne}
\end{figure}

\begin{figure}
\centering
    \includegraphics[scale=0.9]{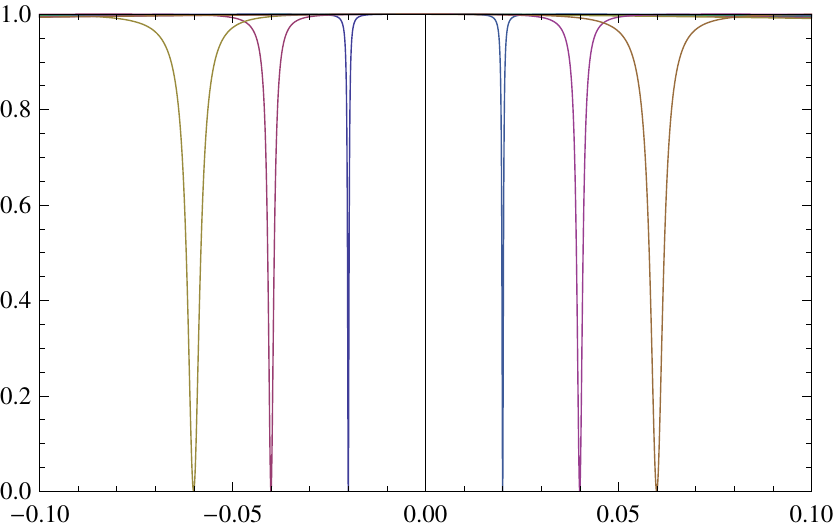}\quad
    \includegraphics[scale=0.9]{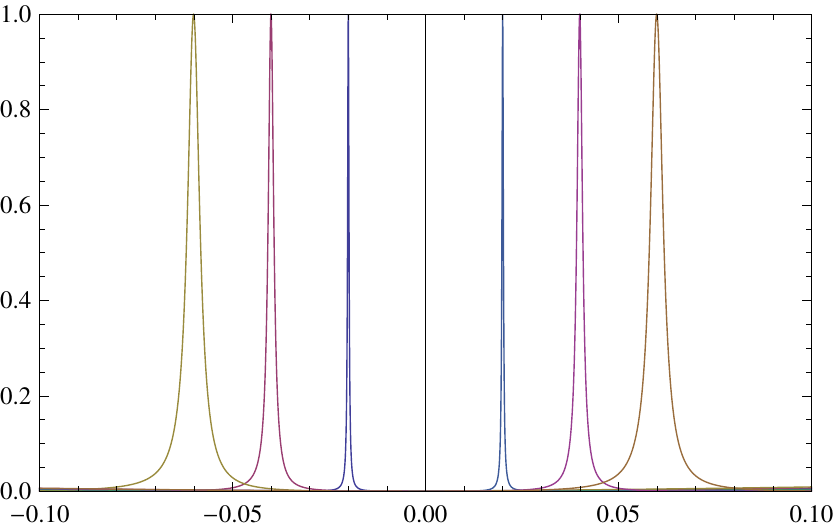}
    \caption{ \small
    $|T|^2$ as a function of $\om$ for $\ka=0,\pm0.01,\pm0.02,\pm0.03$.
      Left: Full background transmission occurs when $\frac{\partial\ell}{\partial\omega}\not=0$, $\frac{\partial a}{\partial\omega}=0$, and $\frac{\partial b}{\partial\omega}\neq0$ at $(\kappa_0,\omega_0)$.  In (\ref{eqn:mixed_expansion}),
       $0<r_1^{(1)}=0.2<t_1=\ell_1=2<r_1^{(2)}=4$,
       $(r_2^{(1)},r_2^{(2)},t_2)=(7, 7, 0.1)$, and $r_0=0.6$.
       For $\ka=0$, $|T|^2\approx1$.
      Right: $\frac{\partial\ell}{\partial\omega}\not=0$, $\frac{\partial a}{\partial\omega}\not=0$, and $\frac{\partial b}{\partial\omega}=0$ at $(\kappa_0,\omega_0)$ and $0<t_1^{(1)}<r_1=\ell_1<t_1^{(2)}$.      
For $\ka=0$, $|T|^2\approx0$. }
    \label{fig:MixRealOne1}
 \end{figure}

\begin{figure}
\centering
    \includegraphics[scale=0.86]{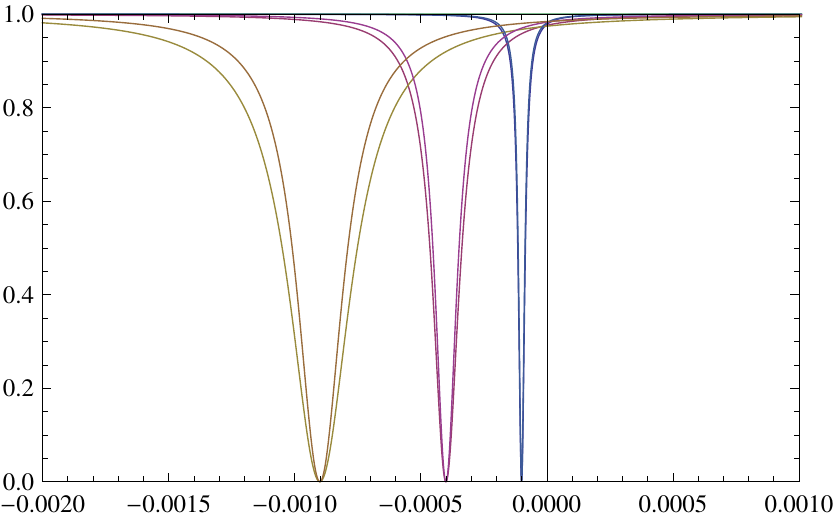}\quad
    \includegraphics[scale=0.90]{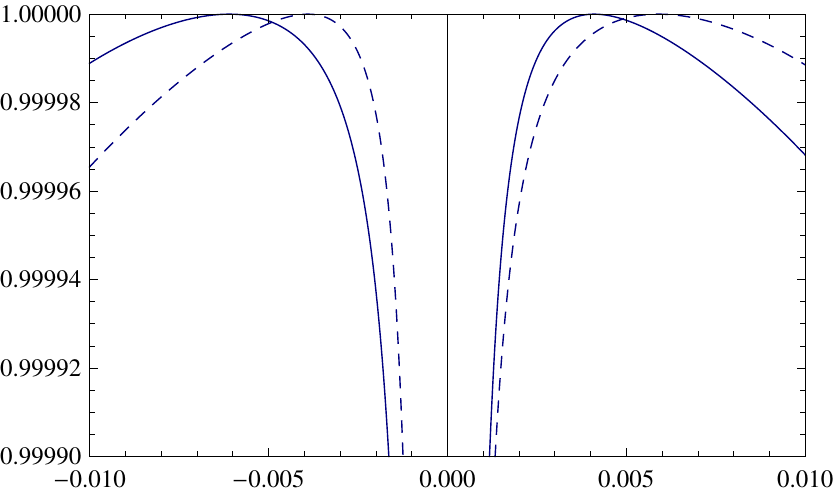}
    \caption{ \small
     $|T|^2$ as a function of $\om$ for $\ka=0,\pm0.01,\pm0.02,\pm0.03$.
Left: Full background transmission occurs when $\frac{\partial\ell}{\partial\omega}\not=0$, $\frac{\partial a}{\partial\omega}=0$, and $\frac{\partial b}{\partial\omega}\neq0$ at $(\kappa_0,\omega_0)$.  In (\ref{eqn:mixed_expansion}),
       $r_1^{(1)}=-0.04<t_1=\ell_1=0<r_1^{(2)}=0.06$;
       $(r_0,r_2^{(1)},r_2^{(2)},t_2)=(0.6,-1, 1, 1)$.
              For $\ka=0$, $|T|^2\approx1$.
      Right: Magnification of the graphs for $\ka=\pm0.01$, bringing into view the frequencies of total transmission.}
    \label{fig:MixRealThree1}
 \end{figure}

\begin{figure}
\centering
    \includegraphics[scale=0.85]{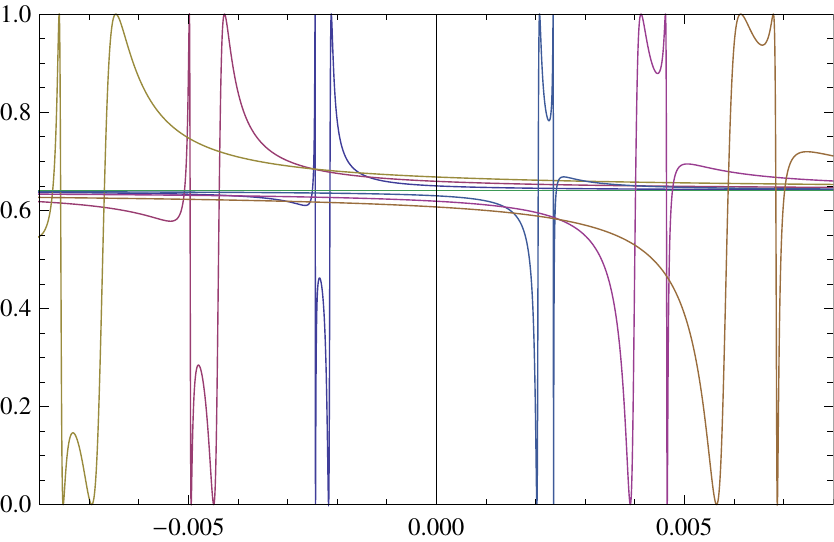} \quad
    \includegraphics[scale=0.855]{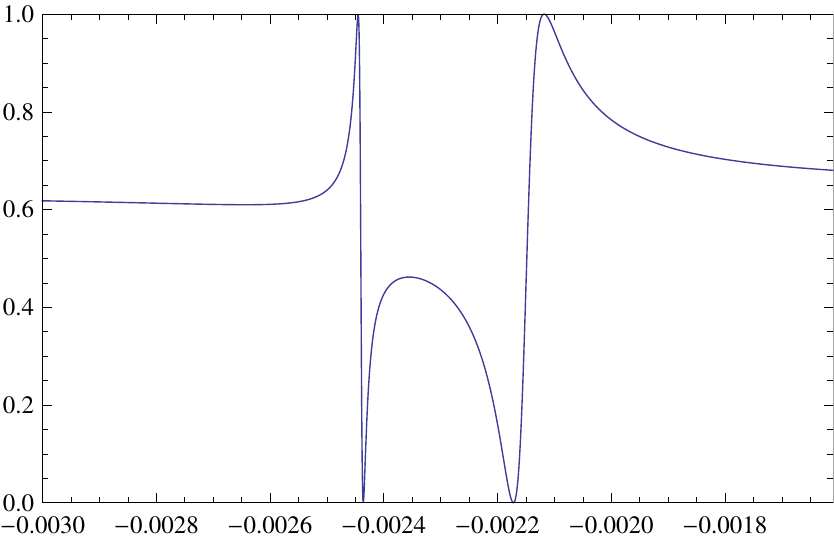}
    \caption{ \small
     Left: $|T|^2$ as a function of $\om$ for $\ka=0,\pm0.003,\pm0.006,\pm0.009$.
     The partial derivatives of $\ell$, $a$, and $b$ all vanish at $\kwz$, whereas their second derivatives are nonzero.
     In (\ref{expansions3}), $(\ell_1^{(1)},\ell_1^{(2)}) = (0.7,0.8)$,
     $(r_0, t_0)=(0.6, 0.8)$,
     $r_2^{(1)}=2<t_2^{(1)}=8$ and $t_2^{(2)}=4<r_2^{(2)}=5$.
     For $\ka=0$, the graph is horizontal, just above $|T|^2=0.6$.
     Right: $\ka=0.003$.   }
    \label{fig:SecondOrderRealOne1}
 \end{figure}
 \begin{figure}
 \centering
    \includegraphics[scale=0.9]{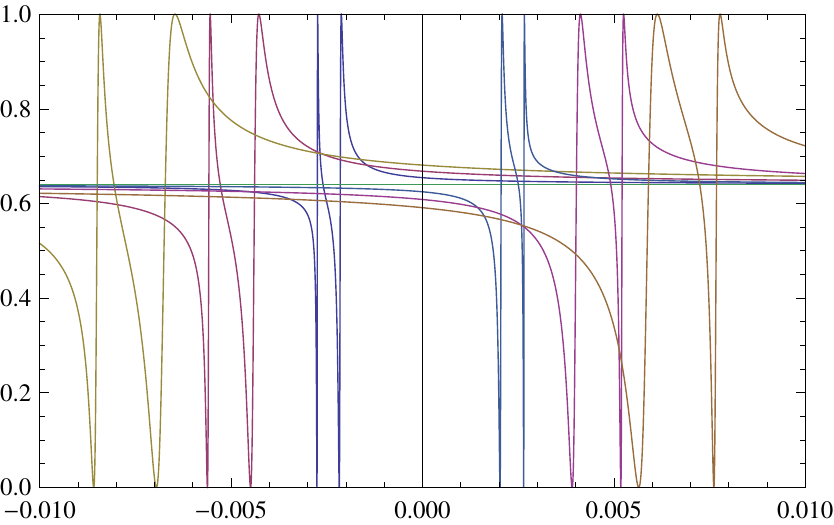}\quad
    \includegraphics[scale=0.877]{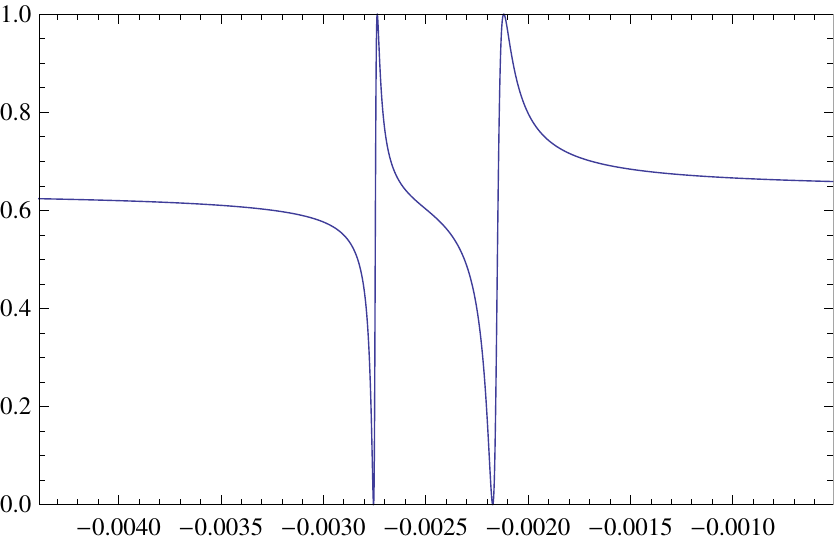}
    \caption{ \small
     Left: $|T|^2$ as a function of $\om$ for $\ka=0,\pm0.003,\pm0.006,\pm0.009$.
     The partial derivatives of $\ell$, $a$, and $b$ all vanish at $\kwz$, whereas their second derivatives are nonzero.
     In (\ref{expansions3}), $(\ell_1^{(1)},\ell_1^{(2)}) = (0.7,0.8)$,
     $(r_0, t_0)=(0.6, 0.8)$,
     $r_2^{(1)}=2<t_2^{(1)}=8$ and $r_2^{(2)}=4<t_2^{(2)}=6$.
          For $\ka=0$, the graph is horizontal, just above $|T|^2=0.6$.
     Right: $\ka=0.003$. }
    \label{fig:SecondOrderRealTwo1}
 \end{figure}

\clearpage

\appendix
\section{Appendix}


\begin{proof}[Proof of Lemma \ref{lemma:analyticityC1}]
To prove $C_1$ is analytic with respect to $\omega$
at $(\kappa_0,\omega_0)$, we let $\omega=\omega_0+\Delta\omega$ and show that
\[
\lim_{\Delta\omega\rightarrow0}
\left\|\frac{C_1(\kappa_0,\omega_0+\Delta\omega)-C_1(\kappa_0,\omega_0)}{\Delta\omega}-C_{1\omega}\right\|=0
\]
in operator norm,
where $C_{1\omega}(\kappa_0,\omega_0):\Honeper\rightarrow \Honeper$ is defined by
\[
(C_{1\omega}u,v)=\frac{1}{\mu_0}
   \sum_m\frac{-i\epsilon_0\mu_0\omega_0}{\sqrt{\ep_0\mu_0\omega_0^2-(m+\kappa_0)^2}}\hat u_m\bar{\hat v}_m,
   \forall u,v\in \Honeper.
\]
The partial derivatives of $\eta_m$ with respect to $\omega$ are
\[
\eta_{m\omega}\kw=\frac{\epsilon_0\mu_0\omega}{\sqrt{\ep_0\mu_0\omega^2-(m+\kappa)^2}},
\]
and the joint analyticity of $\eta_m$ at $(\kappa_0,\omega_0)$ implies that
\[
\begin{split}
\eta_m(\kappa_0,\omega_0+\Delta\omega)&=\eta_m(\kappa_0,\omega_0)+\eta_{m\omega}(\kappa_0,\omega_0)\Delta\omega
  +R^{(2)}_m(\Delta\omega)\\
  &=\eta_m(\kappa_0,\omega_0)+\frac{\epsilon_0\mu_0\omega_0}{\sqrt{\ep_0\mu_0\omega_0^2-(m+\kappa_0)^2}}\Delta\omega
  +R^{(2)}_m(\Delta\omega)
\end{split}
\]
with the remainder term
\[
R^{(2)}_m(\Delta\omega)=\frac{(\Delta\omega)^2}{2\pi i}\int_{{\cal C}_0}\frac{\eta_m(s)}{(s-\omega)(s-\omega_0)^2}ds
\]
in which ${\cal C}_0$ is the circle in the complex plane centered at $\omega_0$ with radius $r_0$ and $2|\Delta\omega|<r_0$.
For any $s\in {\cal C}_0$, $|s-\omega|\ge|r_0-|\Delta\omega||\ge r_0-|\Delta\omega|\ge r_0/2$. We
estimate that for all $m\ne0$, 
\[
\left|R_m^{(2)}(\Delta\omega)\right|\le\frac{|\Delta\omega|^2}{2\pi}\cdot 2\pi r_0\cdot\frac{\sup_{s\in {\cal C}_0}\{|\eta_m(s)|\}}{(r_0-|\Delta\omega|)r_0^2}
 \le \frac{\sup_{s\in {\cal C}_0}\{|\eta_m(s)|\}(|\Delta\omega|)^2}{r_0/2\cdot r_0}.
\]
In the last expression, if $m=0$, $|\eta_m(s)|<C$ for some $C>0$.
If $m\ne0$ and  $\epsilon_0\mu_0\omega_0^2-(m+\kappa)^2\le0$,
$|\eta_m(s)|=\sqrt{(m+\kappa)^2-\epsilon_0\mu_0\omega^2}\le(m+\kappa)\le2m<Cm$ for some $C$;
if $m\ne0$ and $\epsilon_0\mu_0\omega_0^2-(m+\kappa)^2\ge0$, we have
$|\eta_m(s)|\le\sqrt{\epsilon_0\mu_0\omega_0^2}\le C\le Cm$.
To summarize, $|\eta_m(s)|\le C(m+1)$.

For some constants $C',C''>0$,
\[
\begin{split}
&\left|\left(\left[\frac{C_1(\kappa_0,\omega_0+\Delta\omega)-C_1(\kappa_0,\omega_0)}{\Delta\omega}-C_{1\omega}\right]u,v\right)\right|\\
&=\left|\frac{1}{\mu_0}
   \sum_m(-i)
     \left[\frac{(\eta_m(\kappa_0,\omega_0+\Delta\omega)-\eta_m(\kappa_0,\omega_0))}{\Delta\omega}
           -\frac{\epsilon_0\mu_0\omega_0}{\sqrt{\ep_0\mu_0\omega_0^2-(m+\kappa_0)^2}}\right]
     \hat u_m\bar{\hat v}_m  \right|  \\
&=\left|\frac{1}{\mu_0}\sum_m\left[\frac{R^{(2)}_m(\Delta\omega)}{\Delta\omega}\right]\hat u_m\bar{\hat v}_m\right| \le \frac{C}{\mu_0} |\Delta\omega| \left|\sum_m (m+1) \hat u_m\bar{\hat v}_m\right|\\
&\le C'|\Delta\omega| \|u\|_{H^{1/2}(\Gamma)}\|v\|_{H^{1/2}(\Gamma)} \le C'' |\Delta\omega| \|u\|_{\Honeper}\|v\|_{\Honeper}
\end{split}
\]
and
\[
\begin{split}
\lim_{\Delta\omega\rightarrow0}&
\left\|\frac{C_1(\kappa_0,\omega_0+\Delta\omega)-C_1(\kappa_0,\omega_0)}{\Delta\omega}-C_{1\omega}\right\|\\
    &=\lim_{\Delta\omega\rightarrow0}\sup_{u,v\ne0}
    \frac{\left|\left(\left[\frac{C_1(\kappa_0,\omega_0+\Delta\omega)-C_1(\kappa_0,\omega_0)}{\Delta\omega}-C_{1\omega}\right]u,v\right)\right|}
    {\|u\|_{H^1_{per}(\Omega)}\|v\|_{H^1_{per}(\Omega)}}
    =0.
\end{split}
\]
 and hence $C_1$ is analytic with respect to $\omega$.

Now we show that the operator $C_1$ is analytic in $\kappa$.
We prove the limit
\[
\lim_{\Delta\kappa\rightarrow0}\left\|\frac{C_1(\kappa_0+\Delta\kappa,\omega_0)-C_1(\kappa_0,\omega_0)}{\Delta\kappa}
   -C_{1\kappa}\right\|=0
\]
in operator norm, where $C_{1\kappa}$
is defined by
\[
(C_{1\kappa}u,v)=\int_\Omega\frac{1}{\mu}\left[iu\bar{v}_x-iu_x\bar v+2\kappa_0 u\bar v\right]
    +\sum_m\frac{1}{\mu_0}\frac{i(m+\kappa_0)}{\sqrt{\ep_0\mu_0\omega^2_0-(m+\kappa_0)^2}}\hat u_m\bar{\hat v}_m.
\]
The partial derivatives of $\eta_m$ with respect to $\kappa$ are
\[
\eta_{m\kappa}\kw=\frac{-(m+\kappa)}{\sqrt{\epsilon_0\mu_0\omega^2-(m+\kappa)^2}},
\]
and
\[
\eta_m(\kappa_0+\Delta\kappa,\omega_0)-\eta_m(\kappa_0,\omega_0)=
  \frac{-(m+\kappa_0)}{\sqrt{\ep_0\mu_0\omega_0^2-(m+\kappa_0)^2}}\Delta\kappa
   +T_m^{(2)}(\Delta\kappa),
\]
with
$|T^{(2)}_m(\Delta\kappa)|\le D (m+1) |\Delta\kappa|^2$ for some constant $D>0$
and $|\Delta\kappa|$ sufficiently small.
So
\[
\begin{split}
&\left(\left(C_1(\kappa_0+\Delta\kappa,\omega_0)u,v\right)-(C_1(\kappa_0,\omega_0)u,v)\right)\\
&=\int_\Omega\frac{1}{\mu}[\nabla+i(\boldsymbol\kappa_0+\Delta\boldsymbol\kappa)]u
       \cdot[\nabla-i(\boldsymbol\kappa_0+\Delta\boldsymbol\kappa)]\bar v
   -\int_\Omega\frac{1}{\mu}(\nabla+i\boldsymbol\kappa_0)u\cdot(\nabla-i\boldsymbol\kappa_0)\bar v\\
  &\quad +\sum_m\frac{-i}{\mu_0}\left[\eta_m(\kappa_0+\Delta\kappa,\omega_0)
     -\eta_m(\kappa_0,\omega_0)\right]\hat u_m\bar{\hat v}_m\\
&=\int_\Omega\frac{1}{\mu}[\Delta\kappa(iu\bar v_x-iu_x\bar v+2\kappa_0u\bar v)+(\Delta\kappa)^2u\bar v]
  \,+\,\sum_m\frac{-i}{\mu_0}\left[\frac{-(m+\kappa_0)}{\sqrt{\ep_0\mu_0\omega_0^2-(m+\kappa_0)^2}}\Delta\kappa
   +T^{(2)}_m\right]\hat u_m\bar{\hat v}_m
\end{split}
\]
and
\[
\begin{split}
&\left|\left(\frac{\left(C_1(\kappa_0+\Delta\kappa,\omega_0)u,v\right)-\left(C_1(\kappa_0,\omega_0)u,v\right)}{\Delta\kappa}\right)
   -(C_{1\kappa}u,v)\right|\\
   &\quad=\left|\int_\Omega\frac{1}{\mu}u\bar v\Delta\kappa
   +\sum_m \frac{-i}{\mu_0}\frac{T^{(2)}_m}{\Delta\kappa}\hat u_m\bar{\hat v}_m\right|
   \quad\le|\Delta\kappa|\left|\int_\Omega\frac{1}{\mu}u\bar v\right|+\frac{1}{\mu_0|\Delta\kappa|}\left|\sum_m T_m^{(2)}\hat u_m\bar{\hat v}_m\right|\\
   &\quad\le D'|\Delta\kappa|\left(\|u\|_{\Honeper}\|v\|_{\Honeper}+\|u\|_{H^{1/2}(\Gamma)}\|v\|_{H^{1/2}(\Gamma)}\right)
   \quad\le D'' |\Delta\kappa|\|u\|_{\Honeper}\|v\|_{\Honeper}
\end{split}
\]
for some constant $D', D''>0$. Therefore
\[
\begin{split}
\lim_{\Delta\kappa\rightarrow0} &
 \left\|\frac{C_1(\kappa_0+\Delta\kappa,\omega_0)-C_1(\kappa_0,\omega_0)}{\Delta\kappa}
   -C_{1\kappa}\right\| \\
&=\lim_{\Delta\kappa\rightarrow0}\sup_{u,v\in H^1_{per}(\Omega), u,v\ne0}
 \frac{\left|\frac{\left(C_1(\kappa_0+\Delta\kappa,\omega_0)u,v\right)-\left(C_1(\kappa_0,\omega_0)u,v\right)}{\Delta\kappa}
   -(C_{1\kappa}u,v)\right|}{\|u\|_{\Honeper}\|v\|_{\Honeper}} =0
\end{split}
\]
in the operator norm and so $C_1$ is analytic with respect to $\kappa$.

To prove the analyticity of $C_2$ with respect to $\omega$ at $\kwz$, we define an operator $C_{2\omega}\kwz$ by
\[
(C_{2\omega}u,v)_{\Honeper}=-2\omega_0\int_\Omega\ep\, u\bar v
\]
and we have
\[
\begin{split}
&\left( \left[\frac{C_2(\kappa_0,\omega_0+\Delta\omega)-C_2(\kappa_0,\omega_0)}{\Delta\omega}-C_{2\omega}\right]u,v\right)\\
&\quad\quad=\left(-\frac{\omega_0^2+2\omega_0\Delta\omega+\Delta\omega^2-\omega_0^2}{\Delta\omega}+2\omega_0\right)\int_\Omega\epsilon u\bar v
=-\Delta\omega\int_\Omega\epsilon\, u\bar v.
\end{split}
\]
As $\Delta\omega\rightarrow0$, this tends to $0$, and thus $C_2$ is analytic with respect to $\omega$.

The operator $C_2$ does not depend upon $\kappa$.
Because $C_1$ is an analytic automorphism, it has an analytic inverse and, hence, $A$ is analytic.
\end{proof}

\begin{proof}[Proof of Lemma \ref{prop:singlesym_alt}]
We first compare the coefficients in $|\ell|^2-|a|^2-|b|^2=0$ using the expansions (\ref{eqn:single_expansion}) and keeping in mind that $\ell_1=r_1=t_1$.
The coefficients of $\ka^2$, $\ka^2\om$, $\ka\om^2$ and $\om^3$ are
\[
\ell_1
\left[
(\ell_2+\bar \ell_2)+\ell_1(\ell_{\ka}+\ell_{\ka})-r_0^2(r_2+\bar r_2)-\ell_1r_0(r_{\ka}e^{-i\gamma}+\bar r_{\ka}e^{i\gamma})
 -t_0^2(t_2+\bar t_2)+i\ell_1t_0(t_{\ka}e^{-i\gamma}-\bar t_{\ka} e^{i\gamma})
\right],
\]
\[
\begin{split}
(\ell_2+\bar \ell_2)&+2\ell_1(\ell_{\ka}+\bar \ell_{\ka})+\ell_1^2(\ell_{\om}+\bar \ell_{\om})
-r_0^2(r_2+\bar r_2)-2\ell_1r_0(r_{\ka}e^{-i\gamma}+\bar r_{\ka}e^{i\gamma})\\
&-\ell_1^2r_0(r_{\om}e^{-i\gamma}+\bar r_{\om}e^{i\gamma})
-t_0^2(t_2+\bar t_2)-2i\ell_1t_0(t_{\ka}e^{-i\gamma}-\bar t_{\ka}e^{i\gamma})
+i\ell_1^2t_0(t_{\om}e^{-i\gamma}-\bar t_{\om}e^{i\gamma}),
\end{split}
\]
\[
\begin{split}
(\ell_{\ka}+\bar \ell_{\ka})+2\ell_1(\ell_{\om}+\bar \ell_{\om})-r_0(r_{\ka}e^{-i\gamma}
+\bar r_{\ka}e^{i\gamma})-&2\ell_1r_0(r_{\om}e^{-i\gamma}+\bar r_{\om}e^{i\gamma})\\
&+it_0(t_{\ka}e^{-i\gamma}-\bar t_{\ka}e^{i\gamma})
+2i\ell_1t_0(t_{\om}e^{-i\gamma}-\bar t_{\om}e^{i\gamma}),
\end{split}
\]
and
\[
(\ell_{\om}+\bar \ell_{\om})-r_0(r_{\om}e^{-i\gamma}+\bar r_{\om}e^{i\gamma})
   +it_0(t_{\om}e^{-i\gamma}-\bar t_{\om}e^{i\gamma}).
\]
Define the real quantities $A=r_{\ka}e^{-i\gamma}+\bar r_{\ka}e^{i\gamma}$,
$B=i(t_{\ka}e^{-i\gamma}-\bar t_{\ka}e^{i\gamma})$,
$C=r_{\om}e^{-i\gamma}+\bar r_{\om}e^{i\gamma}$,
$D=i(t_{\om}e^{-i\gamma}-\bar t_{\om}e^{i\gamma})$.
Since $\Re(|\ell|^2-|a|^2-|b|^2)=0$, the real parts of the coefficients of $\ka^3$, $\ka^2\om$, $\ka\om^2$, and $\om^3$ are all 0, that is,
\begin{gather*}
\ell_1[2\Re(l_2)+2\ell_1\Re(\ell_{\ka})-2r_0^2\Re(r_2)-\ell_1r_0A-2t_0^2\Re(t_2)+\ell_1t_0B]=0,\\
\begin{split}
2\Re(\ell_2)+2\ell_1\Re(l_{\ka})&+2\ell_1^2\Re(\ell_{\om})-2r_0^2\Re(r_2)-2\ell_1r_0A+\\
&-\ell_1^2r_0C-2t_0^2\Re(t_2)+2\ell_1t_0B+\ell_1^2t_0D=0,
\end{split}\\
2\Re(\ell_{\ka})+4\ell_1\Re(\ell_{\om})-r_0A-2\ell_1r_0C+t_0B+2\ell_1t_0D=0,\\
2\Re(\ell_{\om})-r_0C+t_0D=0.
\end{gather*}
This linear system can be reduced to
\[
\begin{pmatrix}
0 &0&0&0\\
0 &0&0& 0\\
r_0 & t_0 &0 &0\\
0& 0& r_0 & t_0
\end{pmatrix}
\begin{pmatrix}
A\\
B\\
C\\
D
\end{pmatrix}
=
\begin{pmatrix}
\ell_1[2\Re(\ell_2)-2r_0^2\Re(r_2)-2t_0^2\Re(t_2)]\\
2\Re(\ell_2)-2r_0^2\Re(r_2)-2t_0^2\Re(t_2)\\
2\Re(\ell_{\ka})\\
2\Re(\ell_{\om})
\end{pmatrix},
\]
which implies $\Re(\ell_2)-r_0^2\Re(r_2)-t_0^2\Re(t_2)=0$.

The coefficients of $k^3$, $k^2\om$, $k\om^2$, and $\om^3$ in $a\bar b$~are
\begin{gather*}
-i\ell_1r_0t_0(r_2+\bar t_2)+\ell_1\left(r_0\bar t_{\ka}e^{i\gamma}-it_0r_{\ka}e^{-i\gamma}\right),\\
-ir_0t_0(r_2+\bar t_2)+2\ell_1\left(r_0\bar t_{\ka}e^{i\gamma}-it_0r_{\ka}e^{-i\gamma}\right)
   +\ell_1^2\left(r_0\bar t_{\om}e^{i\gamma}-it_0r_{\om}e^{-i\gamma}\right),\\
\left(r_0\bar t_{\ka}e^{i\gamma}-it_0r_{\ka}e^{-i\gamma}\right)
   +2\ell_1\left(r_0\bar t_{\om}e^{i\gamma}-it_0r_{\om}e^{-i\gamma}\right),\\
r_0\bar t_{\om}e^{i\gamma}-it_0r_{\om}e^{-i\gamma}.
\end{gather*}
Since $a\bar b$ is purely imaginary, so are these coefficients, and from the second and third, we obtain
$r_2+\bar t_2\in\mathbb{R}$.

Along the curve $\{(\ka,\om)\in\RR^2:\om+\ell_1\ka=0\}$, the
coefficients of $\ka^4$ in $|\ell|^2-|a|^2-|b|^2$ and in $a\bar b$
are $|\ell_2|^2-r_0^2|r_2|^2-t_0^2|t_2|^2$ and $-ir_0r_2t_0\bar t_2$.  This yields
$|\ell_2|^2=r_0^2|r_2|^2+t_0^2|t_2|^2$ and
$r_2\bar t_2\in\mathbb{R}$.

The relations $r_2+\bar t_2\in\mathbb{R}$ and
$r_2\bar t_2\in\mathbb{R}$ imply that
either $r_2, t_2\in\mathbb{R}$ or $r_2=t_2$.
Because, by Lemma~\ref{lemma:single}, $\ell_2\in\RR\Leftrightarrow r_2=t_2\in\RR$ and because
we assume $\Im(\ell_2)>0$, the numbers $r_2,t_2$ cannot
be identical real numbers.
In the second case with $r_2=t_2\not\in\RR$, from
$\Re(\ell_2)=r_0^2\Re(r_2)+t_0^2\Re(t_2)$,
$|\ell_2|^2=r_0^2|r_2|^2+t_0^2|t_2|^2$, and $r_0^2+t_0^2=1$, we find that
$\Re(\ell_2)=\Re(r_2)=\Re(t_2)$ and $|\Im(\ell_2)|=|\Im(r_2)|=|\Im(t_2)|$.
\end{proof}

\bibliography{ShipmanTu}

\end{document}